\documentclass[journal]{IEEEtran}

\usepackage{graphicx}
\usepackage{subfigure}
\usepackage{hyperref}
\usepackage{colortbl}
\usepackage[table]{xcolor}
\usepackage{bbm}
\usepackage{amsmath}
\usepackage{amsthm}
\usepackage{amssymb}
\usepackage{makecell}
\usepackage{multirow}
\usepackage{algorithm}
\usepackage{algorithmic}
\newtheorem{theorem}{\textbf{Theorem}}

\newtheorem{remark}{\textbf{Remark}}


\begin{document}

\title{Federated Multi-Agent Actor-Critic Learning for Age Sensitive Mobile Edge Computing \thanks{The work was supported by Beijing Natural Science Foundation No. 4202030.

        Zheqi Zhu, Shuo Wan, and Pingyi Fan are with Beijing National Research Center for Information Science and Technology and the Department of Electronic Engineering, Tsinghua University, Beijing 10084, China (e-mail:fpy@tsinghua.edu.cn).

        Khaled B. Letaief is with the Department of ECE, Hong Kong University of Science and Technology, Hong Kong and Pengcheng Lab, Shenzhen 518000, China.}}

% Copyright (c) 2021 IEEE. Personal use of this material is permitted. However, permission to use this material for any other purposes must be obtained from the IEEE by sending a request to pubs-permissions@ieee.org.

\author{Zheqi Zhu, Shuo Wan, Pingyi Fan,~\IEEEmembership{Senior Member,~IEEE}, Khaled B. Letaief,~\IEEEmembership{Fellow,~IEEE}}

% \IEEEauthorblockA{\IEEEauthorrefmark{1}
% Beijing National Research Center for Information Science and Technology and Department of Electronic Engineering Tsinghua University, Beijing 100084, China,
% }}

\maketitle
\IEEEpeerreviewmaketitle
% \thispagestyle{empty} % no page number for the first page
% \pagestyle{empty}  % no page number for the second and the later pages

\begin{abstract}
    As an emerging technique, mobile edge computing (MEC) introduces a new scheme for various distributed communication-computing systems such as industrial Internet of Things (IoT), vehicular communication, smart city, etc. In this work, we mainly focus on the timeliness of the MEC systems where the freshness of the data and computation tasks is significant. Firstly, we formulate a kind of age-sensitive MEC models and define the average age of information (AoI) minimization problems of interests. Then, a novel mixed-policy based multimodal deep reinforcement learning (RL) framework, called heterogeneous multi-agent actor-critic (H-MAAC), is proposed as a paradigm for joint collaboration in the investigated MEC systems, where edge devices and center controller learn the interactive strategies through their own observations. To improve the system performance, we develop the corresponding online algorithm by introducing the edge federated learning mode into the multi-agent cooperation whose advantages on learning convergence can be guaranteed theoretically. To the best of our knowledge, it's the first joint MEC collaboration algorithm that combines the edge federated mode with the multi-agent actor-critic reinforcement learning. Furthermore, we evaluate the proposed approach and compare it with popular RL based methods. As a result, the proposed algorithm not only outperforms the baselines on average system age, but also promotes the stability of training process. Besides, the simulation outcomes provide several insights for collaboration designs over MEC systems.
\end{abstract}


\section{Introduction}
\subsection{Backgrounds}
In the past decades, the number of smart devices has been massively deployed in numerous fields such as industrial IoT, nets of vehicles, environment monitoring networks. As an emerging paradigm for such distributed data systems, edge computing (EC), especially mobile edge computing (MEC), couples the communication with the computation and expands the cloud computing (CC) through allowing the computation to be executed at the edge nodes deployed along the path from data sources to the center cloud. In real applications, an MEC system is composed of three layers. The bottom level contains the sources of the data such as users' equipments (UEs), sensors or web cameras that generate the data in order to provide certain services. At the middle layers, the edge devices can be smartphones, smart routers, intelligent base stations and vehicles with processors on board, which play a relay role. Edge devices can collect data from edge sensors, assist to execute some computation tasks and then relay them to cloud center \cite{shi2016edge}. At the top level, the data and the service requirements, refer to tasks, are assembled at the cloud center. The edge processing schemes need to be designed, so as to improve the system efficiency as well as flexibility \cite{sun2016edgeiot}. That is, it can efficiently utilize the edge computation capabilities to reduce the redundant communication consumption of the networks.

Recently, benefiting from the development of 5G networks, more data can be transmitted with lower latency, which make it possible for MEC systems to be applied in several promising real-time applications. For instance, MEC can support the augment reality (AR), virtual reality (VR), and to enhance the user's experience and promote the quality of the media streaming. In such cases, the images or videos are transmitted and processed in the network and the demands for the data rate and timeliness become stricter than the conventional distributed scenes \cite{hu2015mobile}. Besides, in some urban security scenes such as urgency monitoring, the delay of the data transmission and the processing directly impacts the quality of service \cite{mao2017survey}. To stress out such time-sensitive scenarios, age of information (AoI) \cite{abd2019role} has been introduced to the related investigations as a metric of data freshness.

Generally, the computing resources and capacities of the mobile edge devices are limited. Additionally, the arrival of data as well as the computation tasks are dynamic, and the communication rates between the entities are also constrained by restricted bandwidth. Thus, the scheduling of the MEC systems should be well-designed. To some degree, such a problem is usually cast into a joint stochastic optimization problem with the certain targets. In particular, considering the strategies for the entity operation in the MEC networks, the MEC collaboration shall be fully studied \cite{abbas2017mobile}. The collaborations are mainly considered from the following three aspects: 1) The resource management, including computation resources, power and bandwidth allocation, etc. \cite{liu2020resource,du2018energy}; 2) The edge-level control, i.e., the data collection and trajectory planning for mobile edge devices \cite{8767017}; 3) The data scheduling, including the task execution at the edge devices, the computation assignment for the edge network and the data offloading to the cloud center \cite{mach2017mobile,zhao2019novel}. The critical point of the joint collaboration is to achieve the global optimal strategies for MEC systems. Unfortunately, due to the dynamics of the MEC environment and the complex coupling of the entities from different layers, the related optimization problems are always non-convex and NP-hard \cite{zhao2019computation}. Therefore, the conventional vanilla optimization methods may not work well for the joint MEC collaboration problems and the iterated online approaches with better learning ability and intelligence shall be investigated.

\subsection{Motivations}
Further, for future 6G networks, edge-native artificial intelligence (AI) is regarded to be a potential subject which leverages the MEC systems together with the distributed computing applications. The concepts of \textit{AI for EC} and \textit{EC for AI} are promising visions for future data systems \cite{loven2019edgeai}. In this article, we mainly focus on \textit{AI for EC} and consider how to exploit AI or deep learning (DL) techniques to improve the performance of edge computing systems. Reinforcement learning (RL), especially deep reinforcement learning (DRL) \cite{mnih2015human} has been incorporated into several decision-making scenarios such as automatic robot control and game playing \cite{ye2020mastering}. A widely used class of reinforcement learning approaches are value based RL, e.g. Q-learning and deep Q-network (DQN) which predict the system targets of each action and the make decisions according to an action-value function \cite{sutton2018reinforcement}. For the continuous control or the cases where the action space is extremely large, policy based methods such as DDPG (Deep Deterministic Policy Gradient) are introduced \cite{lillicrap2015continuous}. The basic notion is that an actor module and a critic module are built to fit the best strategies and the evaluation values respectively. Moreover, multi-agent reinforcement learning (MARL) \cite{busoniu2008comprehensive} and federated learning (FL) based reinforcement learning \cite{zhuo2019federated} are also developed for the distributed scenes where agents learn to make decisions through their local observations and cooperate for the same system targets.

Motivated by these, we consider a multi-agent reinforcement learning approach for joint MEC collaboration to maintain the data freshness. In particular, we investigate the cooperative learning framework with mixed policies where the agents learn multiple strategies through local observations and states. Besides, the communication mechanism among learning agents is also a critical point to be studied in this article.

%By exploiting the deep neural networks (DNN) models with high model capacity, DRL are able to learn the latent relations between the states and the actions, which fits the complexity and the dynamics of the environment. The structures of RL framework also evolute. 
\subsection{Related Works}
In the literature, some related researches on the joint collaboration, the age-optimal optimization and the applications of deep learning in MEC systems has been broadly studied.

Ndikumana \textit{et al.} \cite{ndikumana2019joint} proposed a joint framework for communication, task computation, data caching and the distributed control in big data edge computing and evaluated several performance metrics for different procedures. To enhance the mobility and adaptability of the edge computing, MEC systems with unmanned aerial vehicle (UAV) assisted were investigated \cite{merwaday2015uav,sharma2016uav}. The model formulations, including task arrival, computation, data scheduling and communication, were respectively studied in \cite{emara2020spatiotemporal,cao2019intelligent,ning2020partial,matolak2015unmanned}. Zhou \textit{et al.} \cite{zhou2018uav} and Liu \textit{et al.} \cite{liu2019uav} studied the energy efficient joint optimization of UAV-assisted system, and latency aware collaboration were also studied in \cite{liu2017latency}. In terms of the age sensitive MEC system, AoI based metrics are introduced to measure the freshness of data \cite{costa2014age}. Wang  \textit{et al.} \cite{wang2021minimizing} proposed the age of critical information in mobile computing and developed a partially observed scheduling approach. The AoI aware radio resource allocation of multi-vehicular communications was studied in \cite{chen2020age}. Besides, Liu \textit{et al.} \cite{liu2018age} and Hu \textit{et al.} \cite{hu2020aoi} studied the age optimal joint collaboration in time sensitive MEC systems.

A number of researches exploited the multi-stage optimization method or iterative algorithms to solve the joint collaboration problems. However, these existing approaches suffered from the challenges that the real-world scenes are dynamics and complicated, which makes it difficult for these algorithms to extract the latent connections between the environment variation and the entity operation. Hence, some reinforcement learning based online algorithms are developed \cite{tong2020deep} for MEC data systems. The Q-learning and DQN based RL methods were applied for resource allocation \cite{wang2019smart}, task offloading \cite{li2018deep}, as well as trajectory planning \cite{wan2019towards}. Particularly, Chen \textit{et al.} \cite{chen2018optimized} addressed the limitation of Q-learning and proposed a double DQN based offloading algorithm. Additionally, taking the multiple edge devices into consideration, some extended versions of multi-agent reinforcement learning methods were proposed. A multi-agent actor-critic based offloading approach was designed in \cite{wang2020multi1}. Peng \textit{et al.} \cite{peng2020multi} and Wang \textit{et al.} \cite{wang2020multi} adopted multi-agent DDPG (MADDPG) frameworks to resource management and trajectory planning in MEC networks. Besides, in \cite{zhang2020uav}, the vehicular layer MADDPG with attention mechanism was studied for multi-UAV assisted networks. Moreover, by employing federated learning into multi-agent control \cite{kumar2017federated}, Wang \textit{et al.} \cite{wang2020federated} combined the FL and DQN as a decentralized cooperative framework to improve the performance of edge caching.
\subsection{Contributions \& Paper Organization}
The main contributions of this work can be summarized as follows:
\begin{itemize}
    \item[$\bullet$] We put forward the system model and a multi-agent Markov decision process (MDP) formulation to characterize the problems in age sensitive mobile edge computing for further investigations.
    \item[$\bullet$] We build a simulation environment as a \textit{gym} module\footnotemark[1] for these MEC systems, which can be easily employed to test the performance of different collaboration approaches.
    \item[$\bullet$] We present a multi-agent deep reinforcement learning framework, H-MAAC, for MEC joint collaboration. It's a multimodal framework that takes heterogeneous inputs to learn the mixed policies for trajectory planning, data scheduling and resource allocation.
          %It is also an online learning framework where all learning agents observes their local states and train their dual neural nets simultaneously as they interact with the environment.
    \item[$\bullet$] We develop the corresponding multi-agent cooperation algorithm for the online joint collaboration by introducing the edge federated learning mode into the MEC collaboration, abbreviated as EdgeFed H-MAAC\footnote[1]{The codes for the MEC simulation environment and the evaluation of several RL collaboration algorithms are publicly available at \url{https://github.com/Zhuzzq/EdgeFed-MARL-MEC}}, which outperforms DDPG and MADDPG on both system metrics and convergence. To the best of our knowledge, it’s the first joint MEC collaboration algorithm that combines the edge federated mode with the multi-agent actor-critic reinforcement learning.
          %The simulation results demonstrate that the proposed approach outperforms the classical centralized reinforcement learning approaches.
    \item[$\bullet$] The convergence analysis of the proposed MEC collaboration algorithm is also provided in theory which implies that the EdgeFed H-MAAC method reaps better convergence. Besides, the parameter design is also discussed.
          % and implement the proposed algorithms for evaluation. The results show several advantages of H-MAAC based collaboration. In addition, the parameter designs are also discussed.
\end{itemize}

The remaining of the article is organized as follows. In Section \ref{section system} we propose an age sensitive MEC system model and present the problems of interests in this work. In Section \ref{section alg}, we firstly formulate the Markov decision process for the age minimization problem, and then build up a multi-agent edge federated actor-critic learning framework as well as develop the corresponding cooperation algorithm. The learning convergence of the proposed algorithm are presented as a theorem in Section \ref{section convergence}. The simulation results and more discussions are presented in Section \ref{section sim}. Finally, in Section \ref{section conclusion}, we conclude this work and give several potential research directions.

\section{System Model and Problem Formulation}
\label{section system}
In this section, a classic 3-tier multi-agent edge computing system will be firstly introduced, including the corresponding communication and operation models. Then, the problems concerned in age sensitive scenarios will be defined.

\begin{table*}[htbp]
      \caption{\upshape Main notations.}
      \label{notation}
      \newcommand{\tabincell}[2]{\begin{tabular}{@{}#1@{}}#2\end{tabular}}
      \centering
      \begin{tabular}{c|c|l}
            \hline
            \rowcolor{gray!20} & Notations                                                                                                                             & Description                                                                               \\
            \hline
            \multirow{22}{*}{\makecell[c]{Environment                                                                                                                                                                                                              \\Notations}}
                               & $t$                                                                                                                                   & The $t$-th time slot.                                                                     \\
                               & $N\!_s$                                                                                                                               & The number of the data sources.                                                           \\
                               & $N\!_e$                                                                                                                               & The number of the edge devices.                                                           \\
                               & $\mathbb{S}=\left\{S_1,\cdots,S_{N\!_s}\right\}$                                                                                      & The set of data sources.                                                                  \\
                               & $\mathbb{E}=\left\{E_1,\cdots,E_{N\!_e}\right\}$                                                                                      & The set of edge devices.                                                                  \\
                               & $\mathbb{C}$                                                                                                                          & Cloud data center.                                                                        \\
                               & $d_{.}(t)$, $w_{.}(t)$, $idx_{.}(t)$                                                                                                  & \makecell[l]{Attributes of each packet considered in this research, i.e., the data size,  \\the elapsed time and the index of its source.} \\
                               & $D_n(t)=\Big\{\left[d_{n,i}(t),w_{n,i}(t),n\right]\Big\}$                                                                             & $S_n$'s data buffer at $t$-th time slot.                                                  \\
                               & $\boldsymbol{\Delta}(t)=\big\{\Delta_1(t),\cdots,\Delta_{N\!_s}(t)\big\}$                                                             & The list storing the age of all data sources at $t$-th time slot.                         \\

                               & $D^k_{col}(t)=\bigg\{\left[d_{col}^{k,i}(t),w_{col}^{k,i}(t),idx^{k,i}_{col}(t)\right]\bigg\}$                                        & $E_k$'s collected data buffer at caching the data collected but not processed.            \\
                               & $D^k_{exe}(t)=\bigg\{\left[d_{exe}^{k,i}(t),w_{exe}^{k,i}(t),idx^{k,i}_{exe}(t)\right]\bigg\}$                                        & $E_k$'s executed data buffer caching the processed data waiting to be offloaded.          \\
                               & $B_{col}^k$                                                                                                                           & The collected data buffer size of $E_k$.                                                  \\
                               & $B_{exe}^k$                                                                                                                           & The executed data buffer size of $E_k$.                                                   \\
                               & $\boldsymbol{\rm pos}_k(t)=\big[x_k(t),y_k(t),h_k(t)\big]$                                                                            & The position of $E_k$ at $t$-th time slot.                                                \\
                               & $r^k_{move}$, $r^k_{obs}$, $r^k_{collect}$                                                                                            & \makecell[l]{The moving radius (per time slot), observing radius                          \\and the collecting coverage radius of $E_k$.} \\
                               & $\{\boldsymbol{a}_k(t)\}=\Big\{\left[\boldsymbol{\rm move}_k(t),\ \boldsymbol{\rm exe}_k(t),\ \boldsymbol{\rm off}_k(t)\right]\Big\}$ & The action of $E_k$ at $t$, including movement, execution and offloading.                 \\
                               & $\boldsymbol{b}(t)=\left[b_1(t), \cdots, b_{N\!_e\!}(t)\right]$                                                                       & Bandwidth allocation for edge agent offloading at $t$-th time slot.                       \\
            \hline

            \multirow{14}{*}{\makecell[c]{Learning                                                                                                                                                                                                                 \\Framework\\Notations}}& $\mathcal{A}_k$, $\mathcal{A}'_k$                                                            & The actor net and target actor net on $k$-th learning agent.                                                                                         \\
                               & $\mathcal{C}_k$, $\mathcal{C}'_k$                                                                                                     & The critic net and target critic net on $k$-th learning agent.                            \\
                               & $\boldsymbol{\mathcal{A}}=\big\{\mathcal{A}_1,\cdots,\mathcal{A}_{N\!_e},\mathcal{A}_\mathbb{C}\big\}$                                & The set of actor net learning agents for edge devices and center controller.              \\
                               & $\boldsymbol{\mathcal{C}}=\big\{\mathcal{C}_1,\cdots,\mathcal{C}_{N\!_e},\mathcal{C}_\mathbb{C}\big\}$                                & The set of critic net learning agents for edge devices and center controller.             \\
                               & $\boldsymbol{\mathcal{A}'}=\big\{\mathcal{A}'_1,\cdots,\mathcal{A}'_{N\!_e},\mathcal{A}'_\mathbb{C}\big\}$                            & The set of target actor net learning agents corresponding to $\boldsymbol{\mathcal{A}}$.  \\
                               & $\boldsymbol{\mathcal{C}'}=\big\{\mathcal{C}'_1,\cdots,\mathcal{C}'_{N\!_e},\mathcal{C}'_\mathbb{C}\big\}$                            & The set of target critic net learning agents corresponding to $\boldsymbol{\mathcal{C}}$. \\
                               & $\big\{\boldsymbol {s}_k(t)\big\}$, $\boldsymbol s_\mathbb{C}(t)$                                                                     & The input states of edge device and center controller at $t$-th time slot.                \\
                               & $\boldsymbol{\theta}_k$, $\boldsymbol{\theta}'_k$                                                                                     & The parameters of $k$-th learning agent's actor net and target actor net.                 \\
                               & $\boldsymbol{\phi}_k$, $\boldsymbol{\phi}'_k$                                                                                         & The parameters of $k$-th learning agent's actor net and target actor net.                 \\
                               & $\eta_{\mathcal{A}}$, $\eta_{\mathcal{C}}$                                                                                            & Learning rates for actor/critic nets.                                                     \\
                               & $\boldsymbol{\mathcal{B}}$                                                                                                            & Experience buffer for replay.                                                             \\
                               & $\gamma\in[0,1]$                                                                                                                      & Reward/Penalty decay.                                                                     \\
                               & $\tau\in[0,1]$                                                                                                                        & Coefficient for target net updates.                                                       \\
                               & $\omega\in[0,1]$                                                                                                                      & Federated factor for edge updating.                                                       \\

            \hline
      \end{tabular}
\end{table*}

\begin{figure}[htbp]
      \centering
      \begin{minipage}[b]{0.48\textwidth}
            \includegraphics[width=1\textwidth]{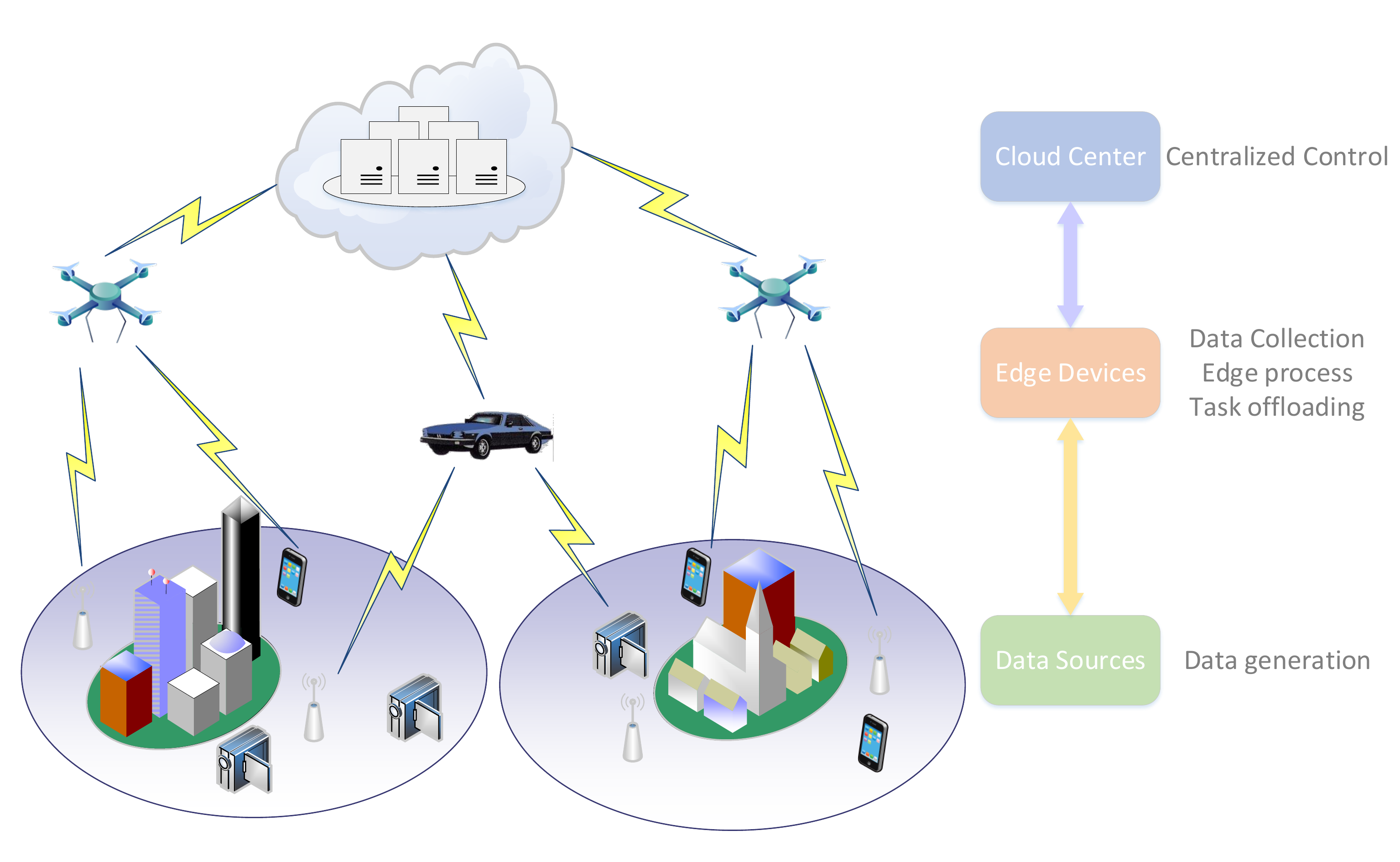}
      \end{minipage}%
      \caption{The 3-Tier Mobile Edge Computing System.}
      \label{system}
\end{figure}
As shown in Fig. \ref{system}, a classic 3-tier edge computing system with data collection is studied. The bottom level are data sources ($\mathbb{S}$) such as sensors, web cameras and users' equipments that continuously generate data packets. Then middle level consists of mobile edge devices ($\mathbb{E}$) such as automobile base stations and unmanned aerial vehicle base stations (UAV-BSs) which manage to move in the area, communicate with the data sources and the cloud center, as well as carry out some processing of the data. At the top level, the cloud center ($\mathbb{C}$) is usually the data center with computing clusters that execute the computing tasks, store the data and implement centralized control for the whole system. The related models of data generation, edge mobility, edge operation, data scheduling, transmission and center controller will be introduced in detail. Additionally, some necessary notations and their explanations are listed in TABLE~\ref{notation}. To explicitly investigate the operations in each time slot, we divide continuous time into small slots and use the integer $t$ to denote the $t$-th time slot.
\subsection{Data Generation Model}
Consider the data sources generate the packets of same formatted data structures, for instance, the sensor data from distributed sensors in IoTs, the image or video data from web cameras in urban security monitoring or virtual/augment reality (VR/AR) scenarios, and the data for a certain class of computing tasks from users with smart equipments. Assuming that the data sources generates independently, we describe the data packets using a tuple of data size, elapsed time and the index of their sources, denoted by $d_{.}(t)$, $w_{.}(t)$, $idx_{.}(t)$ respectively. The arrival packets are temporally stored in the source buffers and wait to be collected.

\subsection{Edge Mobility Model}
In the MEC system introduced above, edge devices are considered to be vehicular base stations with mobility and computing capacity. All edge devices move in the area, collect the packets from data sources, process the data locally and offload the data to cloud center. In the following discussion, we assume that the edge devices $\mathbb E$ are a set of UAVs flying at certain heights over the data sources. Then, the mobility of edge devices can be modelled as:
\begin{equation}
      \label{move}
      \boldsymbol{\rm pos}_k(t+1)=\boldsymbol{\rm pos}_k(t)+\boldsymbol{\rm move}_k(t)
\end{equation}
where $\boldsymbol{\rm pos}_k(t)$ denotes the position of $E_k$ at the beginning of $t$-th time slot, and the movement $\boldsymbol{\rm move}_k(t)$ in each time slot is constrained by
\begin{equation}
      \label{mobility}
      \Big\Vert \boldsymbol{\rm move}_k(t)\Big\Vert_2\leq r^k_{move}
\end{equation}
where $\Vert\cdot\Vert_2$ is the 2-norm of the vectors.

\subsection{Edge Processing \& Data Scheduling Model}
While edge agent $E_k$ hovering over the data source $S_n$, all the packets in $S_n$'s data buffer will be collected and take up one piece of the collected data buffer, $D^k_{col}(t)$. The data pieces in $D^k_{col}(t)$ will be scheduled to be preprocessed using local processor on $E_k$. Since the data packets are assumed to be of formatted structure and the edge preprocess algorithms are determined, the computation on edge devices can be modelled by the data size and the edge computing rate \cite{liu2016delay}. Then, the accumulated execution duration of packets collected from $S_n$ with $E_k$ executing edge processing at the time slot $t$ can be obtained by the equation,
\begin{equation}
      \label{edge execution}
      \tau^k_n(t)=\frac{\sum\limits_i \mathbbm{1}_{\left\{idx^{k,i}_{col}(t)=n\right\}}\cdot d^{k,i}_{col}(t)}{f^k_c(t)}
\end{equation}
where $f^k_c(t)$, related to CPU-cycle frequency, is the $E_k$'s edge execution data rate for the preset tasks at time slot $t$.

On each edge device $E_k$, collected data buffer $D^k_{col}(t)$ cache those data from data sources but have not been edge processed. Assume that in every operation slot, $E_k$ allocates its edge computation resources for one data piece in the buffer. Thus, the edge execution decision at the $t$-th slot can be denoted by a one-hot vector,
\begin{equation}
      \label{exe op}
      \boldsymbol{\rm exe}_k(t)=\big[{\rm exe}_1(t),\cdots,{\rm exe}_{B_{col}^k}(t)\big]
\end{equation}
where ${\rm exe}_i(t)\in\left\{0,1\right\}$ is the CPU allocation flag for each data piece in $E_k$'s collected data buffer and it obviously satisfies the condition,
\begin{equation}
      \sum\limits_{i=1}^{B_{col}^k}{\rm exe}_i(t)=1\ \ \ \ \ for\ k=1,\cdots,N_e.
\end{equation}

After local execution on edge, the data will be cached in executed data buffers, $\big\{D^k_{exe}(t)\big\}$, and wait to be offloaded to the cloud center. Similarly, assume that in every operation slot, $E_k$ decides to offload one piece of packets in $D^k_{exe}(t)$. Thus, the offloading scheduling can be described by a one-hot vector,
\begin{equation}
      \label{off op}
      \boldsymbol{\rm off}_k(t)=\big[{\rm off}_1(t),\cdots,{\rm off}_{B_{exe}^k}(t)\big]
\end{equation}
where ${\rm off}_i(t)\in\left\{0,1\right\}$ are the offloading decisions for each data piece, and they also satisfy the condition,
\begin{equation}
      \sum\limits_{i=1}^{B_{exe}^k}{\rm off}_i(t)=1\ \ \ \ \ for\ k=1,\cdots,N_e.
\end{equation}

\subsection{Communication Model}
There are 3 kinds of transmission links between the entities in the environment, containing source-edge, edge-edge and edge-cloud. In this work, we investigate the cases where edge devices only share the states, observations and learning parameters which means that neither data nor tasks are transmitted between edge devices. Therefore, the transmission costs of edge-edge communication shall be negligible. The transmission process of source-edge and edge-cloud can be modelled as an air to ground (A2G) channel \cite{matolak2015unmanned} where line-of-sight (LoS) path loss as well as none-line-of-sight (NLoS) loss shall be considered \cite{al2014optimal}.
\begin{equation}
      \label{PL}
      PL_\xi(t)=\Big(\frac{4\pi f}{c}\Big)^2\cdot d^2(t)\cdot\eta_\xi
\end{equation}
where $d(t)=\sqrt{x^2(t)+y^2(t)+h^2(t)}$ is the distance between the edge device and the ground entity (a chosen data sources or cloud center), $f$ is the carrier frequency and $c$ is the speed of light. Besides, $\eta_\xi$ with $\xi=\{0,1\}$ represents the excessive path loss of LoS and NLoS cases.
Hence, the average A2G path loss of the communication channel for $E_k$-$S_n$ (or $E_k$-$\mathbb{C}$) at $t$-th slot can be obtained by
\begin{equation}
      \label{path loss}
      \overline{L}_{k,n}(t)=p_0(t)\cdot PL^{k,n}_0(t)+p_1(t)\cdot PL^{k,n}_1(t)
\end{equation}
where $p_0(t)$, $p_1(t)$ are the probability of LoS and NLoS which can be closely approximated by the following form:
\begin{equation}
      \label{p_los}
      p_0(t)=\frac{1}{1+a\exp\big(-b(\psi-a)\big)}
\end{equation}
where $\psi=\tan^{-1}\Big(\frac{h(t)}{\sqrt{x^(t)+y^2(t)}}\Big)$ is the angle between the edge-ground link and the horizontal plane. Moreover, $a$ and $b$ are parameters related to the environment. Then, considering the frequency division mode with total bandwidth $W$, for the channel with allocated bandwidth proportion $b_{k,n}(t)$, the transmission rate between $E_k$ and $S_n$ (or $\mathbb{C}$) can be expressed by
\begin{equation}
      \label{trans rate}
      R_{k,n}(t)=b_{k,n}(t)W\log_2\Big(1+\frac{P^k_{tr}(t)}{\overline{L}_{k,n}(t)N_0b_{k,n}(t)W}\Big)
\end{equation}
where $N_0$ is the noise power spectral density and $P^k_{tr}(t)$ represents the power for transmission satisfying
\begin{equation}
      \label{power}
      0\leq P^k_{tr}(t)\leq P^k_{tr,max}.
\end{equation}
\subsection{Problem Formulation}
In the age sensitive scenarios, the freshness of the data shall be significantly stressed. Recapping the concept, age of information (AoI) \cite{kosta2017age}, in the system introduced above, the age of data source $S_n$ at $t$-th time slot is defined as the subtraction of current time and the generation time of the latest data at the receiver \cite{costa2016age}, which can be expressed as:
\begin{equation}
      \label{age}
      \Delta_n(t)=t-T_g^n(t)
\end{equation}
where $T_g^n(t)$ denotes the generation time of $S_n$'s latest data packet received by the cloud center. Different from the delay of each data packet, $\Delta_n(t)$ is a duration measure for each data source which implies how frequently the data of $S_n$ are collected, edge executed and offloaded, which is especially important to such time-sensitive MEC system. Thus, an $N_s$-dimension vector $\boldsymbol{\Delta}(t)$ is used to record the age of each data source. To maintain the timeliness, the target is to minimize the average age of all data sources by controlling the edge devices, scheduling the data and allocating the resources effectively. Therefore, taking the above system models into account as the constraints, we obtain the following optimization problem:
\begin{align}
      \label{P1}
      \mathcal P_1:\quad & \min\limits_{\{\boldsymbol{a}_k(t)\},\boldsymbol{b}(t)}\ \ \Bigg[\overline{\boldsymbol{\Delta}}(t):=\frac{1}{N_s}\sum\limits_{n=1}^{N_s}\Delta_n(t)\Bigg] \\
                         & \begin{array}{l@{\ }l@{}l@{\ }l}
            \notag
            \mbox{s.t.}\quad & \quad\ Eq.(\ref{move})\ \ to\ \ Eq.(\ref{power}). \\
      \end{array}
\end{align}

Before solving the optimization problem above, let us rethink the introduced system models and corresponding problems from following three perspectives. Firstly, the constraints of $\mathcal P_1$ are heterogeneous and the optimization objectives are of two stages (edge stage and cloud stage). Thus, the optimal solutions cannot be explicitly expressed and an iterative algorithm shall be adopted. Secondly, note that $\mathcal P_1$ is an instant optimization problem and the environment sates are stochastic, which means that the strategies for each time slot should be variant. Moreover, because $\Delta_n(t)$ depends on not only current states but also the states of previous time, the optimal solutions for each time slot may not lead to full-time optimization. For above reasons, we are to formulate the optimization problem into MDP game and discuss the reinforcement learning based solutions as \cite{wiering2012reinforcement}.

\section{An Edge Federated Actor-Critic Learning Framework for Multi-Agent Cooperation}
\label{section alg}
In this section, a Markov decision process (MDP) will be modelled for optimization problem $\mathcal{P}_1$. And then, with the MDP formulation for MEC collaboration, we shall adopt multi-agent reinforcement learning approaches to solve the MDP problem \cite{van2012reinforcement}. Q-Learning and DQN are popular value based RL methods which learn the action-value function $Q(\boldsymbol{s},\boldsymbol{a})$ related to system reward/penalty. However, while the action spaces grow too large, the search for optimal actions becomes extremely hard. To overcome the complexity of the action space, policy based methods such as A2C and DDPG are introduced, where dual neural networks are employed to estimate the action $\boldsymbol{a}$ and $Q$ value respectively. Moreover, in multi-agent scenes, the single learning agent mode or centralized RL requests a large neural network with a complex structure and massive model parameters, which may suffer from some difficulties on training convergence and model generalization \cite{littman1994markov,panait2005cooperative}. For above reasons, we present a heterogeneous multi-agent actor critic learning framework (H-MAAC) for MEC collaboration and the corresponding algorithm to solve the optimization problem. Besides, the convergence analysis of the proposed algorithm will be given.
\subsection{MDP Formulation for MEC Collaboration}
Markov decision process is a common model to formulate such environment-interactive systems \cite{puterman2014markov}. An MDP game can be expressed by a tuple with four elements, $\mathcal{M}\big\{S,A,R,\mathcal{T}\big\}$, standing for the states, actions, rewards and transition policies respectively. Note that in optimization problem $\mathcal{P}_1$, the target is to minimize the objective average AoI. Therefore, we rewrite the element $R$ as $P$ to represent the penalty of the game and obtain the substitute tuple of MDP formulation, $\mathcal{M}\big\{S,A,P,\mathcal{T}\big\}$. Furthermore, a multi-agent extended version of MDP contains a number of agents to match the scenarios with multiple controllable entities. In the above MEC collaboration system, all mobile edge devices and the center controller can be regarded as the agents. Overall we have $N_e$ edge agents and one center agent. All the agents observe their states and act with certain strategies.
\subsubsection{States}
For edge agents, their states $\{\boldsymbol{s}_k(t)\}$ contain the local observations of the environment, the status of the edge devices including buffer states, allocated offloading bandwidth, etc. The states of center agent, $\boldsymbol{s}_\mathbb{C}(t)$, consists of all edge device status.
\subsubsection{Actions}
In the cases of interest, edge devices move, collect data, locally execute and offload tasks to cloud center. Then, the edge agents' actions are composed of movement, execution decision and offloading scheduling, denoted by
\begin{equation}
      \big\{\boldsymbol{a}_k(t)\big\}=\Big\{\left[\boldsymbol{\rm move}_k(t),\ \boldsymbol{\rm exe}_k(t),\ \boldsymbol{\rm off}_k(t)\right]\Big\}
\end{equation}
Meanwhile, the cloud center controller allocates the offloading bandwidth for each edge device. Thus, the action of center agent, $\mathrm{a}_\mathbb{C}(t)$, is the one-sum bandwidth proportion vector $\boldsymbol{b}(t)$.
\subsubsection{Penalties}
Since we focus on the edge collaboration in age sensitive MEC system where all agents collaborate to minimize the average age of data sources, the global penalty will be shared for all agents. Then, the current penalty at $t$-th slot for each agent can be described as
\begin{equation}
      \label{penalty}
      p_k(t)=\overline{\Delta}(t)\ \ for\ k=1,\cdots,N_e,\mathbb{C}.
\end{equation}
To investigate the global optimization of the system, the following long-time penalty with decay is considered:
\begin{equation}
      \label{long-time penalty}
      P_k(t)=\sum\limits_{i=0}^{T}\gamma^i\cdot p_k(t+i)
\end{equation}
where $\gamma$ is the decay coefficient and $T$ is the length of time window.
\subsubsection{Transition Policies}
For MEC system, it is hard to find a formatted policy to cover all the state transitions of data sources, edge devices, cloud center as well as resource allocation. As a result, we use
\begin{equation}
      \label{transition}
      \mathcal{T}\Big(\{\boldsymbol{s}_k(t+1)\},\boldsymbol{s}_\mathbb{C}(t+1)\Big|\{\boldsymbol{s}_k(t)\},\boldsymbol{s}_\mathbb{C}(t),\{\boldsymbol{a}_k(t)\},\boldsymbol{a}_\mathbb{C}(t)\Big)
\end{equation}
to represent the entities' interactions in the system.
\subsection{An Edge-Federated Heterogeneous Multi-Agent Actor-Critic Framework}
\begin{figure}[htbp]
      \centering
      \begin{minipage}[b]{0.48\textwidth}
            \includegraphics[width=1\textwidth]{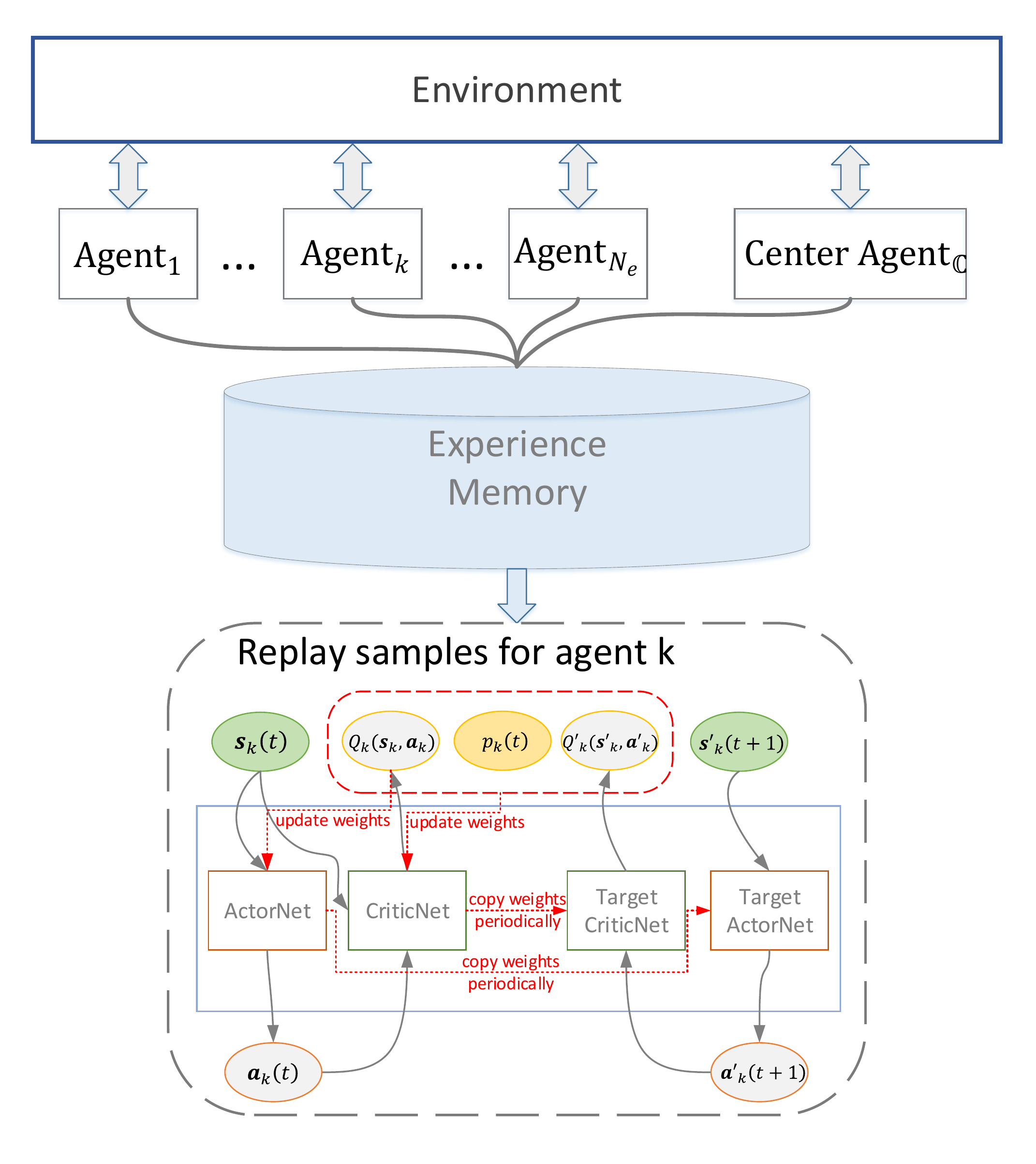}
      \end{minipage}%
      \caption{Architecture of the proposed H-MAAC based reinforcement learning collaboration framework}
      \label{learning framework}
\end{figure}
\begin{figure*}[htbp]
      \centering
      \subfigure[Edge Actor Net.]{
            \label{aanet} %% label for first subfigure
            \begin{minipage}{0.48\linewidth}
                  \includegraphics[width=1\textwidth]{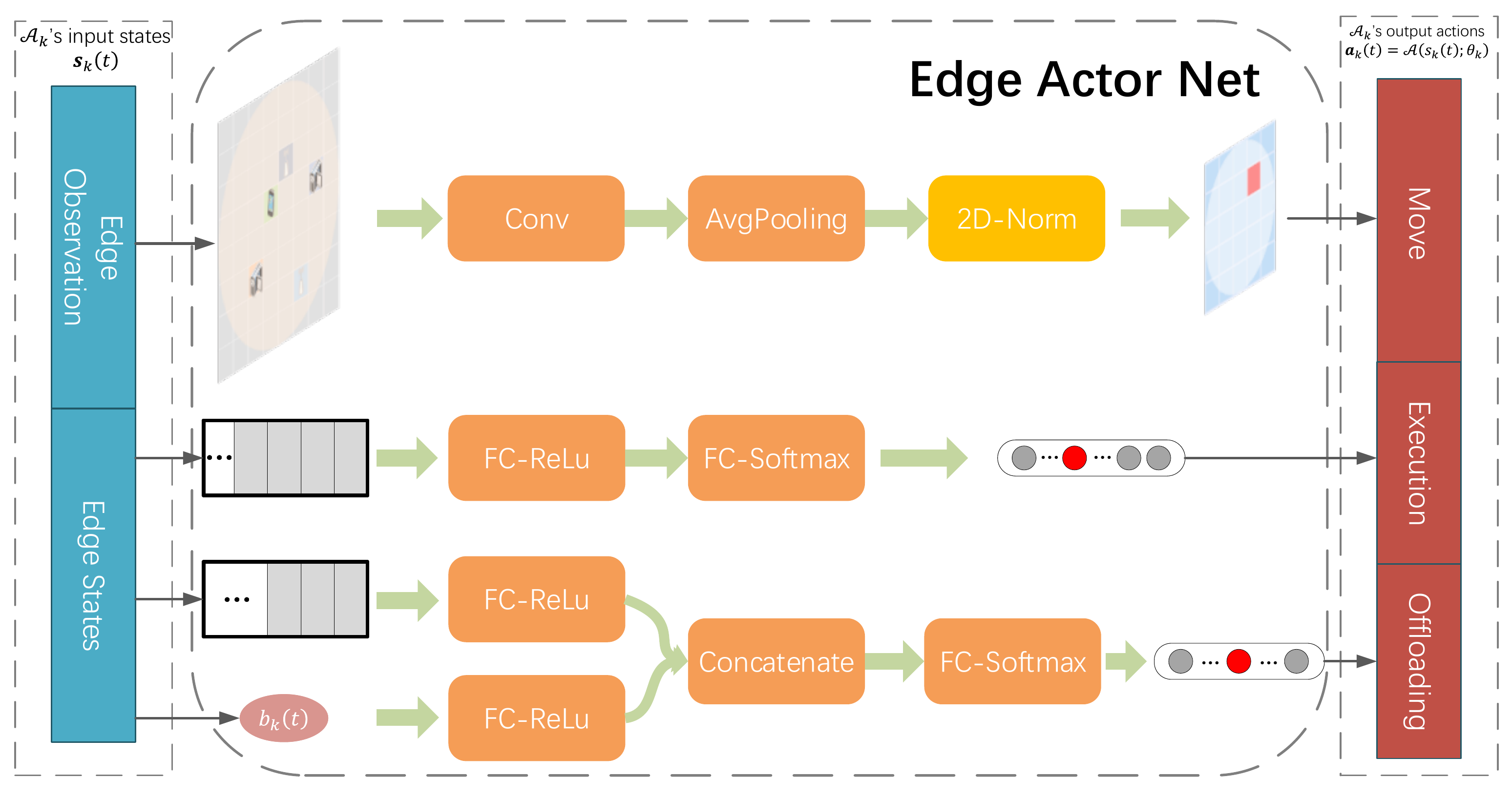}
            \end{minipage}}%
      \subfigure[Center Actor Net.]{
            \label{canet} %% label for second subfigure
            \begin{minipage}{0.48\linewidth}
                  \includegraphics[width=1\textwidth]{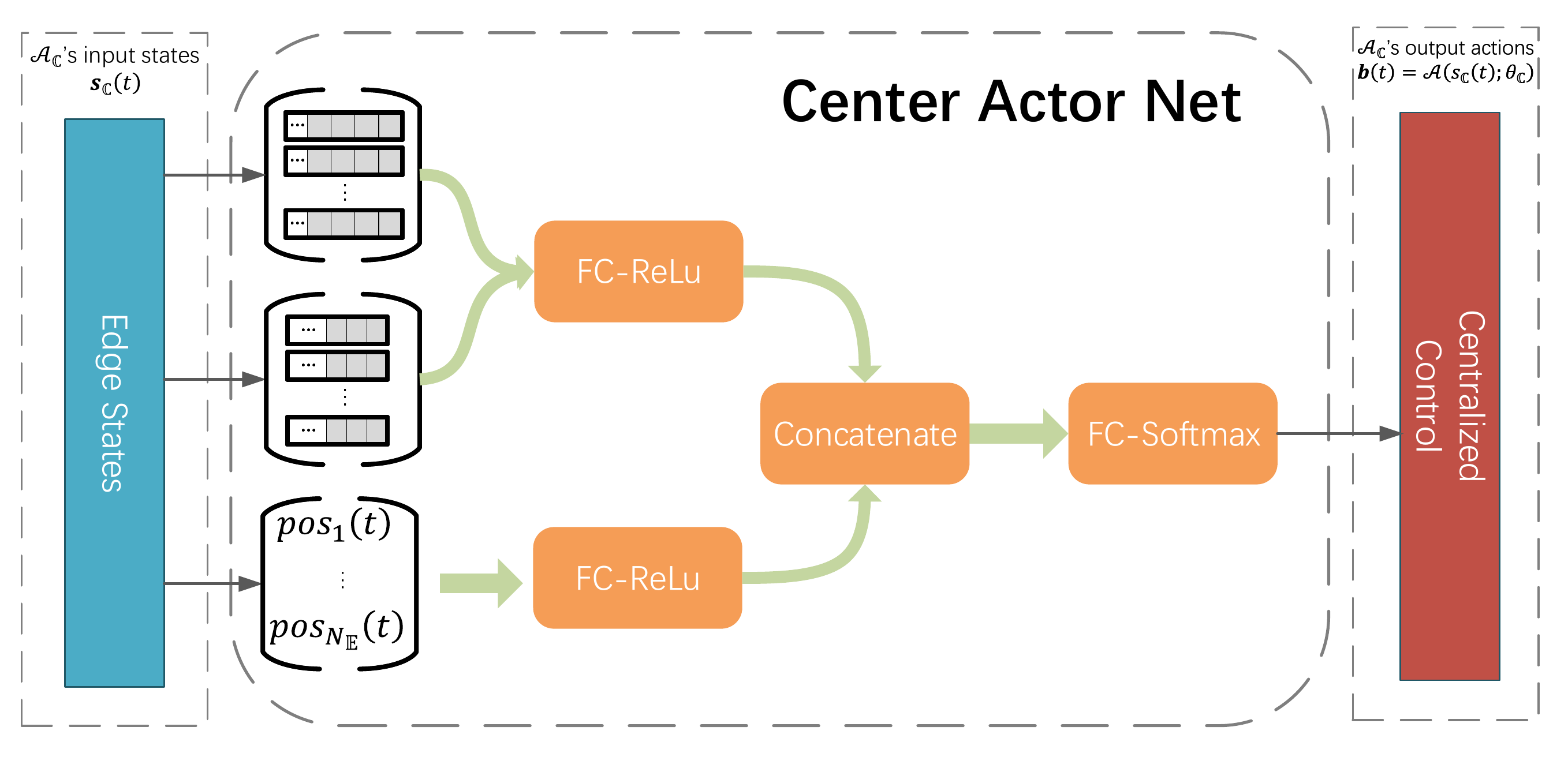}
            \end{minipage}}

      \subfigure[Edge Critic Net.]{
            \label{acnet} %% label for second subfigure
            \begin{minipage}{0.48\linewidth}
                  \includegraphics[width=1\textwidth]{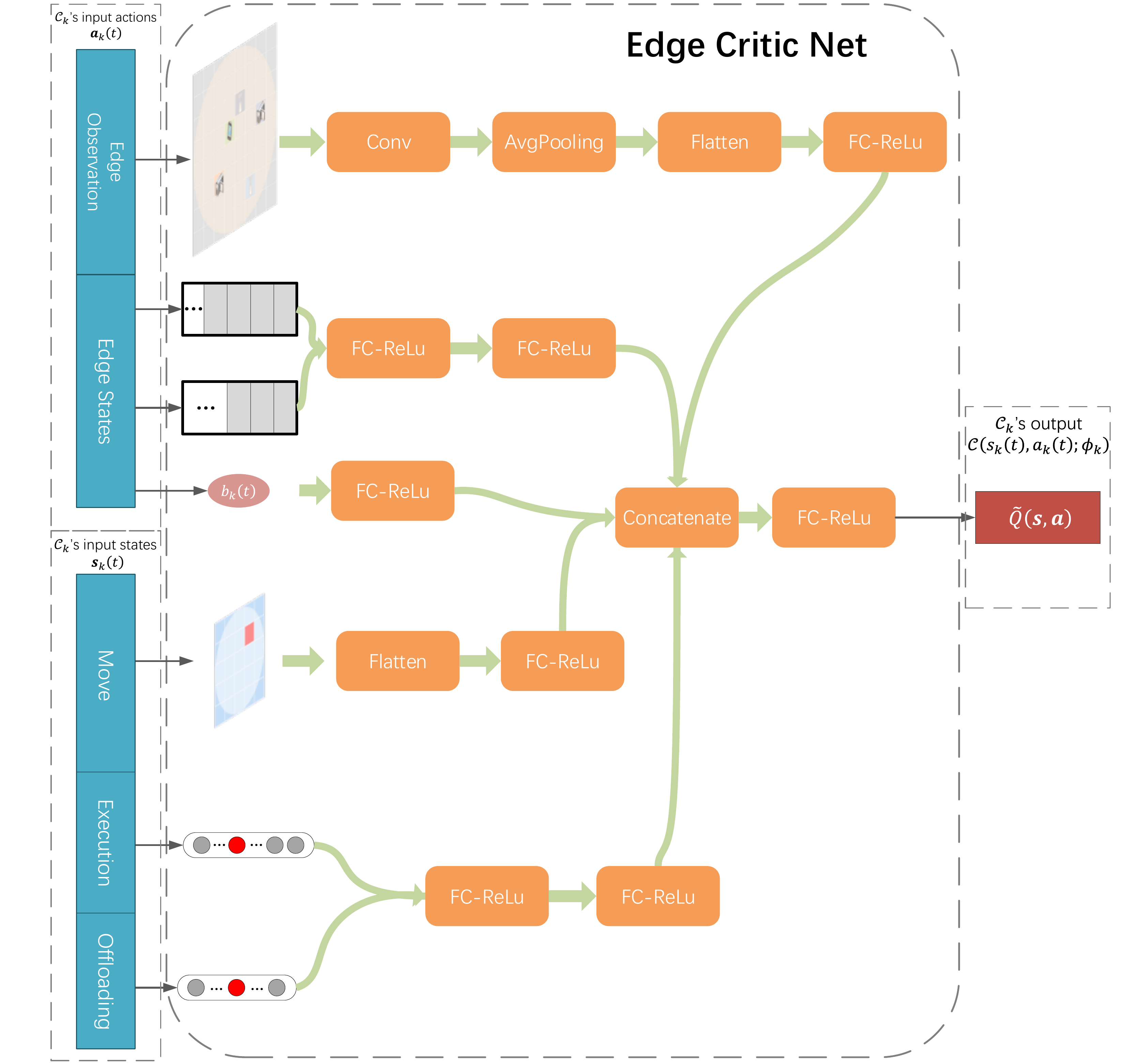}
            \end{minipage}}
      \subfigure[Center Critic Net.]{
            \label{ccnet} %% label for second subfigure
            \begin{minipage}{0.48\linewidth}
                  \includegraphics[width=1\textwidth]{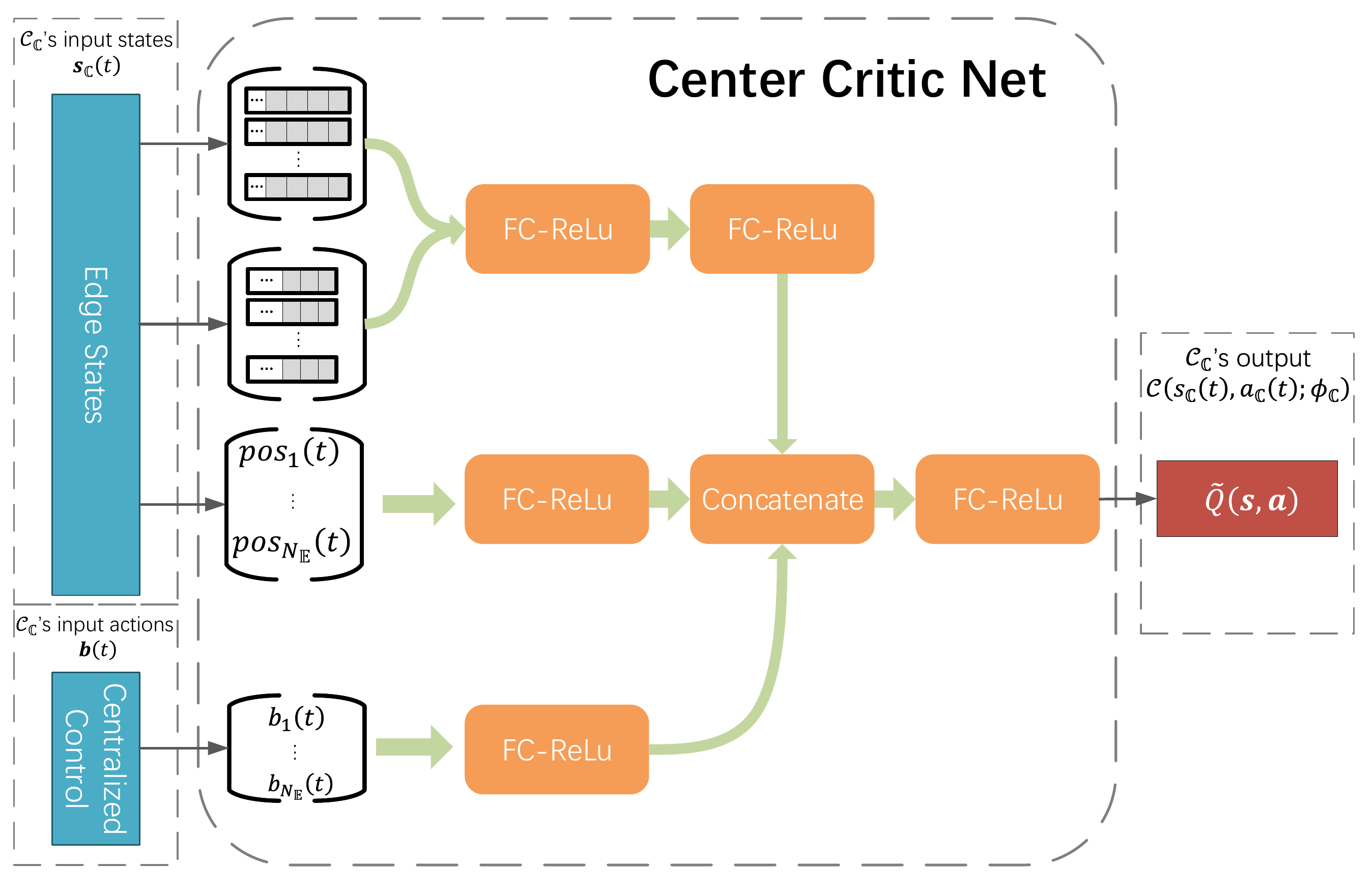}
            \end{minipage}}
      \caption{Neural network models of each actor-critic agent.}
      \label{Net Model} %% label for entire figure
\end{figure*}
To control the edge devices and optimize the center bandwidth allocation, we present a heterogeneous multi-agent actor-critic (H-MAAC) framework. The construction of the framework and the learning procedures are shown in Fig. \ref{learning framework}, where learning agents interact with the environment, memorize the experience replay and learn the optimal actions to minimize the system penalty \cite{lowe2017multi}. The heterogeneity comes from three aspects: 1) The multi-modality of input states; 2) The multiple output actions; 3) The ensemble neural network models to learn the mixed policies. Overall, dual lightweight neural networks are built for each learning, containing original actor/critic nets and target actor/critic nets. Since the structures of the networks are relatively simple which will cost less computation resources, the framework can be deployed on each edge devices. The details of the framework are as follows.
\subsubsection{Neural network design}
Actor nets, $\mathcal{A}_k\big(\boldsymbol{s}_k(t);\boldsymbol{\theta}_k\big)$, take the states of each agent as input and output the current actions $\boldsymbol{a}_k(t)$. Since the input states and the output actions are multimodal, the network should be designed according to the specific data structures. For edge agents, we build a multiple input-output neural network for each device to learn the diversified actions through the ensemble states, including the local observation for data sources, the edge buffer states and the offloading channel states. The structure of the edge agent actor net is shown in Fig. \ref{aanet} where a convolutional neural network (CNN) with average pooling is employed to obtain the movement action while multilayer perceptrons (MLPs) are employed for edge execution as well as offloading scheduling. Specifically, we format the local observation for data sources into an $r^k_{obs}\times r^k_{obs}\times 2$ map. The observation maps can be regarded as 2-channel images where the third dimension refers to the aggregated size and delay of the data packets sensed by edge device $E_k$. Inspired by the successful applications of CNNs in computer visions, we adopt a CNN based neural network to extract the area with larger data packets as well as higher AoI. Then, through the average pooling, we project the $r^k_{obs}\times r^k_{obs}\times 2$ observation map onto the $r^k_{move}\times r^k_{move}\times 1$ movement map to decide the trajectory actions. The other inputs, buffer states and allocated bandwidth are formatted as lists and scalars. And the data scheduling vectors are outputted by an MLP net. For center actor net, its inputs consist of the state lists of the edge devices and the output is a one-sum vector for bandwidth allocation. As for the center agent actor, we use an MLP to combine the multiple edge state lists to allocate bandwidth proportion for edge-center communication, as shown in Fig. \ref{acnet}.
Additionally, on each learning agent, a critic net, $\mathcal{C}_k\big(\boldsymbol{s}_k(t),\boldsymbol{a}_k(t);\boldsymbol{\phi}_k\big)$, is also deployed to approximate the action-value function $Q\big(\boldsymbol{s}_k(t),\boldsymbol{a}_k(t)\big)$ with the current states and actions as inputs. Note that the objectives of all $\mathcal{C}_k$ are consistent in order to minimize the average age of the whole system. The critic nets are designed by the main structures of actor nets plugging the layers for action evaluation, as shown in Fig. \ref{canet} and Fig. \ref{ccnet}.
\subsubsection{Target net update}
In addition to the original actor-critic nets, the target actor-critic nets are also built. Target nets have the same structure and initialization as the original nets. While training the network parameters, target nets estimate the future actions $\boldsymbol{a}'_k(t+1)$ as well as $Q'\big(\boldsymbol{s}'_k(t+1),\boldsymbol{a}'_k(t+1)\big)$ values based on the states of next slot. The employ of target nets improves the stability and convergence of replay training \cite{arulkumaran2017brief}. The parameters of target nets are slowly updated by the original nets every $T_u$ period with the mixing weight $\tau$, i.e.,
\begin{equation}
      \label{target update}
      \begin{aligned}
            \boldsymbol{\theta}'_k=\tau\boldsymbol{\theta}'_k+(1-\tau)\boldsymbol{\theta}_k \\
            \boldsymbol{\phi}'_k=\tau\boldsymbol{\phi}'_k+(1-\tau)\boldsymbol{\phi}_k.
      \end{aligned}
\end{equation}
\subsubsection{Experience replay}
To improve the online learning efficiency, an experience replay approach is exploited here. The interactions between the entities and the environment, denoted as tuples, $\left\{\boldsymbol{s}(t),\boldsymbol{a}(t),\overline{\Delta}(t),\boldsymbol{s}'(t+1)\right\}$, are stored in an experience buffer $\boldsymbol{\mathcal{B}}$. The experience buffer is of finite capacity. When the buffer is full, new records will replace the oldest ones. As for the sampling rules, to guarantee synchronization of learning and the environment interactions, the latest $B/2$ data are always exploited for training. In each learning epoch, each agent samples interaction data with batch size $B$ from the experience buffer and updates the network parameters with the learning agents $\eta_{\mathcal{A}}$ and $\eta_{\mathcal{C}}$. Specifically, the critic nets are updated by minimizing the MSE loss function,
\begin{equation}
      \label{C loss}
      \ell_{\mathcal{C}_k}(\boldsymbol{\phi}_k):=\mathrm{E}\left[\Vert\mathcal{C}_k(\boldsymbol{s}_k,\boldsymbol{a}_k;\boldsymbol{\phi}_k)-\hat{y_k}\Vert^2\right]
\end{equation}
where
\begin{equation}
      \label{yhat}
      \hat{y}_k=\overline{\Delta}+\gamma\mathcal{C}'_k(\boldsymbol{s}'_k,\boldsymbol{a}'_k;\boldsymbol{\phi}'_k)
\end{equation}
is the estimated long-time $Q$ value. Since the goal is to minimize the penalty, the loss function of actor nets can be defined as the predicted $Q$ value,
\begin{equation}
      \label{A loss}
      \ell_{\mathcal{A}_k}({\boldsymbol\theta}_k):=\mathcal{C}_k\Big(\boldsymbol{s}_k,\mathcal{A}_k(\boldsymbol{s}_k;{\boldsymbol{\theta}}_k);\boldsymbol{\phi}_k\Big).
\end{equation}
Then, we summarize the training procedure at $t$-th epoch as Algorithm \ref{alg:replay}.
\begin{algorithm}
      \caption{Experience Replay Procedure}
      \label{alg:replay}
      \begin{algorithmic}[1]
            \FOR{each agent $k$ in $\left\{1,\cdots,N_e,\mathbb{C}\right\}$}
            \STATE Sample $\left\{\boldsymbol{s}_k,\boldsymbol{a}_k,\overline{\Delta},\boldsymbol{s}'_k\right\}$ from $\boldsymbol{\mathcal{B}}[k]$;\\
            \STATE Predict new actions: $\boldsymbol{a}'_k=\mathcal{A}'_k(\boldsymbol{s}'_k;\boldsymbol{\theta}'_k)$;\\
            \STATE Predict new $Q$-value: $Q'(\boldsymbol{s}'_k,\boldsymbol{a}'_k)=\mathcal{C}'_k(\boldsymbol{s}'_k,\boldsymbol{a}'_k;\boldsymbol{\phi}'_k)$;\\
            \STATE Calculate $\hat{y}_k$ by Eq.(\ref{yhat});\\
            \STATE Calculate $\ell_{\mathcal{C}_k}(\boldsymbol{\phi}_k)$, $\ell_{\mathcal{A}_k}(\boldsymbol{\theta}_k)$ by Eq.(\ref{C loss}), Eq.(\ref{A loss});\\
            \STATE Update network parameters using SGD optimizer:\\
            \quad\quad\quad\quad\quad $\boldsymbol{\phi}^{t+1}_k\leftarrow\boldsymbol{\phi}^t_k-\eta_{\mathcal{C}}\nabla_\phi\tilde{\ell}_{\mathcal{C}_k}(\boldsymbol{\phi}^t_k)$,\\
            \quad\quad\quad\quad\quad $\boldsymbol{\theta}^{t+1}_k\leftarrow\boldsymbol{\theta}^t_k-\eta_{\mathcal{A}}\nabla_\theta\tilde{\ell}_{\mathcal{A}_k}(\boldsymbol{\theta}^t_k)$.\\

            \ENDFOR
      \end{algorithmic}
\end{algorithm}

\subsubsection{Edge-federated mode}
Note that it is a cooperative model and the penalties of edge agents are identical. Commonly speaking, so as to reach the global optimum, the cross-communication is required in such multi-agent learning scenarios to share the knowledge of different agents. However, in the MEC system of interests, the encoding of the multimodal input states is hard to design. What's more, the transmission and the processing of the observation will cost excessive communication and computation resources. Hence, to overcome these difficulties, inspired by the concept of federated learning \cite{yang2019federated}, we propose an edge-federated mode for the above framework, where every $E_f$ learning epoch, all edge agents share their actor net parameters and carry out the federated updating. Under the proposed updating rule, each edge agent preserves the parameters with weight $\omega$ and mixes the others' parameters, which can be formulated by
\begin{equation}
      \label{ef updating}
      \boldsymbol{\theta}^{t+1}=\boldsymbol{\theta}^t\cdot\boldsymbol{\Omega}
\end{equation}
where $\boldsymbol{\theta}^t=\left[\boldsymbol{\theta}^t_1,\cdots,\boldsymbol{\theta}^t_{N_e}\right]$ denotes the vector of all edge actor nets at the $t$-th learning epoch and $\boldsymbol{\Omega}$ denotes the federated updating matrix,
\begin{equation}
      \label{Omega}
      \boldsymbol{\Omega}=
      \left[ \begin{array}{cccc}
                  \omega                 & \frac{1-\omega}{N_e-1} & \cdots & \frac{1-\omega}{N_e-1} \\
                  \frac{1-\omega}{N_e-1} & \omega                 & \cdots & \frac{1-\omega}{N_e-1} \\
                  \vdots                 & \vdots                 & \ddots & \vdots                 \\
                  \frac{1-\omega}{N_e-1} & \frac{1-\omega}{N_e-1} & \cdots & \omega
            \end{array}
            \right ]
\end{equation}
On the one hand, edge-federated provide the model-wise communication for each edge agent. Instead of sharing the input states, only the parameters of the lightweight actor nets are transmitted, which improves the communication efficiency of the system \cite{konevcny2016federated} and such communication costs can be neglected compared to the data volume in the considered MEC system. On the other hand, the edge-federated mode performs better learning convergence which will be discussed in Section \ref{section convergence}.
\subsubsection{Exploitation-exploration}
Online learning suffers from the exploitation-exploration dilemma, namely, the agents tend to repeat the previous actions which may cause trap at some position and the loss of exploration in the MEC system. To avoid such phenomenon during the online interactions with the MEC environment, we adopt an $\epsilon$-exploration approach \cite{sutton2018reinforcement} which enforces random actions with the probability $\epsilon$.

\subsection{Online Learning Algorithms for Multi-Agent Cooperation}
\begin{algorithm}[h]
      \caption{EdgeFed H-MAAC Online Collaboration}
      {\bf Initialization:} Initialize system parameters and hyper parameters for learning.\\
      {\bf Initialization:} Initialize net parameters $\boldsymbol{\theta}_k$, $\boldsymbol{\phi}_k$ and set target nets: $\boldsymbol{\theta}'_k\leftarrow\boldsymbol{\theta}_k$, $\boldsymbol{\phi}'_k\leftarrow\boldsymbol{\phi}_k$.\\
      \label{alg:all}
      \begin{algorithmic}[1]
            \FOR{epoch $t$ = 1 to MAX\_EPOCH}
            \STATE Randomly generate $q\in[0,1]$;
            \FOR{each agent $k$ in $\left\{1,\cdots,N_e,\mathbb{C}\right\}$}
            \IF{$p<\epsilon$ or $\vert\boldsymbol{\mathcal{B}}[k]\vert<B$}
            \STATE Randomly choose actions $\boldsymbol{a}_k(t)$;
            \ELSE
            \STATE Ensemble local observation and states: $\boldsymbol{s}_k(t)$;
            \STATE Set actions: $\boldsymbol{a}_k(t)=\mathcal{A}_k\left(\boldsymbol{s}_k(t);\boldsymbol{\theta}_k\right)$;
            \ENDIF
            \ENDFOR
            \STATE Interact with environment and obtain $\overline{\Delta}(t)$, $\boldsymbol{s}'(t+1)$;
            \STATE Add $\left\{\boldsymbol{s},\boldsymbol{a},\overline{\Delta},\boldsymbol{s}'\right\}$ into $\boldsymbol{\mathcal{B}}$;
            \FOR{each agent $k$ in $\left\{1,\cdots,N_e,\mathbb{C}\right\}$}
            \IF{$\vert\boldsymbol{\mathcal{B}}[k]\vert\geq B$}
            \STATE Run replay procedure, Algorithm \ref{alg:replay};
            \ENDIF
            \ENDFOR
            \IF{$t \mod T_u==1$}
            \STATE Update target nets, as Eq.(\ref{target update});
            \ENDIF
            \IF{$t \mod E_f==1$}
            \STATE Run edge-federated updating, as Eq.(\ref{ef updating});
            \ENDIF
            \ENDFOR
      \end{algorithmic}
\end{algorithm}
Based on the proposed learning framework, we develop the corresponding online multi-agent collaboration as Algorithm \ref{alg:all} where the agents learn and update the optimal strategies while the MEC system works continuously. More specifically, each epoch comprises four procedures: (1)Line 1 to line 11 is the acting procedure with $\epsilon$-exploration, where agents choose whether to act randomly or to follow the actor net strategies; (2) Line 13 to line 17 is the replay training procedure of the networks which will be skipped if the count of samples in $\boldsymbol{\mathcal{B}}$ is insufficient; (3) Line 18 to line 20 is the periodical target net updating procedure; (4) Line 21 to line 23 is the edge-federated updating procedure.

\section{Convergence Analysis}
\label{section convergence}
In this section, we shall investigate the convergence of the collaboration algorithms and show that the edge-federated learning mode has a better performance in terms of convergence rate and stability, compared to the original H-MAAC.

Firstly, consider the convergence of actor nets, the objective function is defined by the average loss of all edge agents' actor nets, i.e.,
\begin{equation}
      \label{objective l}
      \mathcal L(\overline{\boldsymbol \theta})=\frac{1}{N_e}\sum\limits_{k=1}^{N_e}\ell_{\mathcal{A}_k}(\overline{\boldsymbol\theta}_k)
\end{equation}
where $\ell_{\mathcal{A}_k}(\overline{\boldsymbol\theta}_k):=\mathcal{C}_k\Big(\boldsymbol{s}_k,\mathcal{A}_k(\boldsymbol{s}_k;\overline{\boldsymbol{\theta}}_k);\boldsymbol{\phi}_k\Big)$ is the loss of $\mathcal{A}_k$ for edge agent $E_k$ and $\overline{\boldsymbol\theta}_k$ represents the updated model parameters of $\mathcal{A}_k$ under the federated rule, Eq.(\ref{ef updating}).

Secondly, since the proposed algorithm is online and the environment is dynamic, the training may not guarantee either the global optimal or the penalty stability. Hence, we investigate the learning convergence from an alternative perspective, the gradients' 2-norm time average of all edge actor nets,
\begin{equation}
      \label{gradient 2-norm}
      \frac{1}{N_eT}\sum\limits_{t=T_0+1}^{T_0+T}\sum\limits_{k=1}^{N_e}\Big\Vert\nabla\ell_{\mathcal{A}_k}(\overline{\boldsymbol\theta}^t_k)\Big\Vert^2
\end{equation}
where $T$ is the time horizon length after $T_0$-th learning epoch. Additionally, we denote the training interaction sets of $E_k$ from epoch $T_0+1$ to $T_0+T$ as
\begin{equation}
      \label{training set}
      \{\boldsymbol{\mathcal{T}}_k\}^{T_0}:=\Big\{\boldsymbol{\mathcal{T}}_k(T_0+1),\cdots,\boldsymbol{\mathcal{T}}_k(T_0+T)\Big\}
\end{equation}
where $\boldsymbol{\mathcal{T}}_k(t)=\left\{\boldsymbol{s}_k,\boldsymbol{a}_k,\overline{\Delta},\boldsymbol{s}'_k\right\}^t$ represents the interaction records used to train $E_k$'s neural networks at $t$-th learning epoch.

\subsection{Assumptions}
Before the analysis, let us illustrate some assumptions similar to \cite{wang2018cooperative}, under which the convergence theorem can be carried out.
\subsubsection{\textbf{Existence of optimal loss}}
While the environment state sets are given, we assume that the optimal net parameters, $\boldsymbol{\theta}^*_k$, and the corresponding actor loss, $\ell^*_{\boldsymbol{A}_k}$, of each edge agent exist. This assumption is intuitive because the action space is finite and the loss is related to system penalty.
\subsubsection{\textbf{Independence of center agent learning}}
Since we consider the federated learning mode of edge agents, to eliminate the influence of center agent, we assume that the training process of center agent is independent with the edge agents'. Consequently, the learning of center actor-critic parameters, $\boldsymbol{\theta}_\mathbb{C}$ and $\boldsymbol{\phi}_\mathbb{C}$ can always reach the optimal parameters and be out of account while discussing the edge learning convergence.
\subsubsection{\textbf{Fine-fitness of critic nets}}
Note that each critic net is trained to predict the common penalty value $Q$ related to the system average age, $\overline{\boldsymbol{\Delta}}(t)$. Owing to the large learning capacity of neural networks, we assume that after $T_0$-th epoch, for given training set $\{\boldsymbol{\mathcal{T}}_k\}^{T_0}$, critic nets converge and fit the $Q$ values well, which refers to that the critic nets' parameters $\boldsymbol{\phi}_k$ keep the consistent values of $\boldsymbol{\phi}^{T_0*}_k$ and the loss $\ell_{\mathcal{A}_k}(\boldsymbol\theta_k)=\mathcal{C}_k\Big(\boldsymbol{s}_k,\mathcal{A}_k(\boldsymbol{s}_k;\boldsymbol{\theta}_k);\boldsymbol{\phi}_k^*\Big)$ can be regarded as a determined function with fixed $\boldsymbol{\phi}_k^*$.
\subsubsection{\textbf{Conditional L-smoothness for given environment states}}
In general, neural networks are neither smooth nor convex. However, for given training set $\{\boldsymbol{\mathcal{T}}_k\}^{T_0}$, the Lipschitz constant of a model consisting of MLP or CNN can be estimated \cite{balan2017lipschitz}. Besides, the activation functions used in proposed net such as ReLU and Softmax are Lipschitz continuous and differentiable \cite{latorre2020lipschitz}. Therefore, we can denote the Lipschitz constant of each actor net as $L_{\boldsymbol{\mathcal{T}}_k}$ which is related to its inputs $\{\boldsymbol{\mathcal{T}}_k\}^{T_0}$ from $T_0+1$-th to $T_0+T$-th training epoch. Then, the conditional L-smoothness of each actor net's loss function can be expressed as
\begin{equation}
      \label{l-smooth}
      \Vert\nabla\ell_{\mathcal{A}_k}(\boldsymbol\theta_k)-\nabla\ell_{\mathcal{A}_k}(\boldsymbol\theta'_k)\Vert\leq L_{\boldsymbol{\mathcal{T}}_k}\Vert\boldsymbol{\theta}_k-\boldsymbol{\theta}'_k\Vert.
\end{equation}
\subsubsection{\textbf{Unbiased bounded SGD}}
In each epoch, each agent samples a mini batch with size $B$ from the experience memory. Consider an SGD (stochastic gradient descent) based optimizer is applied for back propagation, the unbiased stochastic gradient $\tilde{\boldsymbol{g}}_k$'s variance is bounded by
\begin{equation}
      \label{sgd bound}
      \mathrm{E}\big[\Vert\tilde{\boldsymbol g}_k-\boldsymbol g_k\Vert^2\big]\leq C\Vert \boldsymbol g_k\Vert^2+\frac{\sigma_{\boldsymbol{\mathcal{T}}_k}^2}{B}
\end{equation}
where $\boldsymbol g_k=\mathrm{E}\big[\tilde{\boldsymbol g}_k\big]=\nabla\ell_{\mathcal{A}_k}(\boldsymbol\theta_k)$ is the average stochastic gradient for given training input states $\{\boldsymbol{\mathcal{T}}_k\}^{T_0}$ and $\sigma_{\boldsymbol{\mathcal{T}}_k}$ is the related variance constant. Moreover, $C$ is the non-negative constant for all edge agents.

Briefly, the assumptions 1) to 3) are reasonable to the MEC environment while 4) and 5) are assumptions for the learning progress.
\subsection{Convergence Theorem for EdgeFed H-MAAC}
Under the assumptions above, the learning convergence can be conducted and the results can be summarized as the following theorem.
\begin{theorem}
      \label{convergence thm}
      For the proposed EdgeFed H-MAAC algorithm with the update period $E_f$, and under the assumptions 1) to 5), if the learning rates of all edge agents are set to be $\eta$ which satisfies
      \begin{equation}
            \begin{aligned}
                  \Bigg[\frac{2CE_fL_{max}^2}{1-\zeta^2} & +\frac{E_f^2L_{max}^2}{1-\zeta}\left(\frac{2\zeta}{1+\zeta}+\frac{2\zeta}{1-\zeta}+\frac{E_f-1}{E_f}\right)\Bigg]\eta^2 \\
                                                         & +L_{max}(C+1)\eta-1\leq 0
            \end{aligned}
            \label{lr}
      \end{equation}
      where $\zeta=\frac{N_e\omega-1}{N_e-1}$ is the second maximum eigenvalue of the updating matrix $\boldsymbol{\Omega}$ and $L_{max}=\max\limits_k\ L_{\boldsymbol{s}_k}$. Then, the time average squared gradient norm after $T_0$-th epoch is bounded by
      \begin{equation}
            \begin{aligned}
                  \mathrm{E} \Bigg[\frac{1}{N_eT} & \sum\limits_{t=T_0+1}^{T_0+T}\sum\limits_{k=1}^{N_e}\big\Vert\nabla\ell_{\mathcal{A}_k}(\overline{\boldsymbol{\theta}}^t_k)\big\Vert^2\Bigg]\leq \frac{2\sum\limits_{k=1}^{N_e}\left[\ell_{\mathcal{A}_k}(\overline{\boldsymbol{\theta}}^{T_0}_k)-\ell_{\mathcal{A}_k}^*\right]}{\eta N_eT} \\
                                                  & +\frac{\eta}{N_eB}\sum\limits_{k=1}^{N_e}L_{\boldsymbol{\mathcal{T}}_k}\sigma_{\boldsymbol{\mathcal{T}}_k}^2+\frac{\eta^2\sigma_{max}^2L_{max}^2}{2B}\Big(\frac{1+\zeta^2}{1-\zeta^2}E_f-1\Big).
            \end{aligned}
            \label{thm}
      \end{equation}
\end{theorem}
\begin{proof}
      \label{proof}
      We present the proof of the theorem in the way similar to Appendix D of \cite{haddadpour2019convergence}. Consider the difference of the average loss,
      \begin{equation}
            \mathcal{L}(\overline{\boldsymbol{\theta}}^{t+1})-\mathcal{L}(\overline{\boldsymbol{\theta}}^{t})=\frac{1}{N_e}\sum\limits_{k=1}^{N_e}\Big[\ell_{\mathcal{A}_k}(\overline{\boldsymbol\theta}_k^{t+1})-\ell_{\mathcal{A}_k}(\overline{\boldsymbol\theta}_k^t)\Big].
            \label{def minus}
      \end{equation}
      As an application of the L-smoothness gradient assumption, we obtain
      \begin{equation}
            \begin{aligned}
                  \ell_{\mathcal{A}_k}(\overline{\boldsymbol\theta}_k^{t+1})-\ell_{\mathcal{A}_k}(\overline{\boldsymbol\theta}_k^t) \leq \frac{\eta^2L_{\boldsymbol{\mathcal{T}}_k}}{2}\big\Vert\tilde{\boldsymbol{g}}^t_k\big\Vert^2-\eta\left<\nabla\ell_{\mathcal{A}_k}(\overline{\boldsymbol\theta}_k^t),\tilde{\boldsymbol g}_k^t\right>.
            \end{aligned}
            \label{part1}
      \end{equation}
      By assumption 5), the expected value of the second term can be bounded by
      \begin{equation}
            \begin{aligned}
                  \mathrm{E}\Big[\Vert\tilde{\boldsymbol{g}}^t_k\Vert^2\Big] & =\mathrm{E}\Big[\Vert\tilde{\boldsymbol{g}}^t_k-\boldsymbol{g}^t_k\Vert^2\Big]+\Vert\boldsymbol{g}^t_k\Vert^2 \\
                                                                             & \leq(C+1)\Vert\boldsymbol{g}^t_k\Vert+\frac{\sigma_{\boldsymbol{\mathcal{T}}_k}^2}{B}.
            \end{aligned}
            \label{g bound}
      \end{equation}
      For the mean of the first term, by Eq.(\ref{l-smooth}) in assumption 4), we have
      \begin{equation}
            \begin{aligned}
                  \mathrm{E} & \Big[-\eta\left<\nabla\ell_{\mathcal{A}_k}(\overline{\boldsymbol\theta}_k^t),\tilde{\boldsymbol g}_k^t\right>\Big] =-\eta\left<\nabla\ell_{\mathcal{A}_k}(\overline{\boldsymbol\theta}_k^t),\mathrm{E}\left[\tilde{\boldsymbol g}_k^t\right]\right> \\
                             & =-\eta\left<\nabla\ell_{\mathcal{A}_k}(\overline{\boldsymbol\theta}_k^t),\boldsymbol g_k^t\right>                                                                                                                                                   \\
                             & =\frac{-\eta}{2}\left[\Vert\nabla\ell_{\mathcal{A}_k}(\overline{\boldsymbol\theta}_k^t)\Vert^2+\Vert\boldsymbol g_k^t\Vert^2-\Vert\nabla\ell_{\mathcal{A}_k}(\overline{\boldsymbol\theta}_k^t)-\boldsymbol g_k^t\Vert\right]                        \\
                             & \leq\frac{\eta}{2}\left[-\Vert\nabla\ell_{\mathcal{A}_k}(\overline{\boldsymbol\theta}_k^t)\Vert^2-\Vert\boldsymbol g_k^t\Vert^2+L_{\boldsymbol{\mathcal{T}}_k}^2\Vert\overline{\boldsymbol{\theta}}^t_k-\boldsymbol{\theta}^t_k\Vert^2\right]
            \end{aligned}
            \label{product}
      \end{equation}
      According to Eq.(136) of Appendix D.2.4 in \cite{wang2018cooperative}, we get the average bound for the last term in Eq.(\ref{product}) as
      \begin{equation}
            \begin{aligned}
                  \frac{1}{N_eT}\sum\limits_{t,k} & L^2_{\boldsymbol{\mathcal{T}}_k}\Vert\overline{\boldsymbol{\theta}}^t_k-\boldsymbol{\theta}^t_k\Vert^2 \leq\frac{\eta^2\sigma_{max}^2L^2_{max}}{B}\left(\frac{1+\zeta^2}{1-\zeta^2}E_f-1\right) \\
                                                  & +\Bigg[\frac{\eta^2E_f^2L^2_{max}}{1-\zeta}\left(\frac{2\zeta}{1+\zeta}+\frac{2\zeta}{1-\zeta}+\frac{E_f-1}{E_f}\right)                                                                         \\
                                                  & +\frac{2\eta^2CE_fL^2_{max}}{1-\zeta^2}\Bigg]\frac{1}{N_eT}\sum\limits_t\sum\limits_k \Vert\boldsymbol g_k^t\Vert^2
            \end{aligned}
            \label{part2}
      \end{equation}
      Then, by taking the average on both sides of Eq.(\ref{def minus}) and combining Eq.(\ref{part1}) to Eq.(\ref{part2}), we obtain
      \begin{equation}
            \begin{aligned}
                  \frac{1}{T} & \sum\limits_{t=T_0+1}^{T_0+T}\mathrm{E}\left[\mathcal{L}(\overline{\boldsymbol{\theta}}^{t+1})-\mathcal{L}(\overline{\boldsymbol{\theta}}^{t})\right]\leq\frac{1}{N_eT}\sum\limits_t\sum\limits_k\frac{\eta^2L_{\boldsymbol{\mathcal{T}}_k}}{2}\big\Vert\tilde{\boldsymbol{g}}^t_k\big\Vert^2 \\
                              & \quad\quad\quad +\frac{1}{N_eT}\sum\limits_t\sum\limits_k-\eta\left<\nabla\ell_{\mathcal{A}_k}(\overline{\boldsymbol\theta}_k^t),\tilde{\boldsymbol g}_k^t\right>                                                                                                                             \\
                              & \leq\frac{-\eta}{2N_eT}\sum\limits_t\sum\limits_k\Vert\nabla\ell_{\mathcal{A}_k}(\overline{\boldsymbol\theta}_k^t)\Vert^2+\boldsymbol{\Gamma}\cdot\frac{\eta}{2N_eT}\sum\limits_t\sum\limits_k\Vert\boldsymbol{g}^t_k\Vert^2                                                                  \\
                              & +\frac{\eta^3\sigma_{max}^2L_{max}^2}{2B}\left(\frac{1+\zeta^2}{1-\zeta^2}E_f-1\right)+\frac{\eta^2}{2N_eB}\sum\limits_kL_{\boldsymbol{\mathcal{T}}_k}\sigma_{\boldsymbol{\mathcal{T}}_k}^2
            \end{aligned}
      \end{equation}
      where $\Gamma$ is the left-hand side of Eq.(\ref{lr}),
      \begin{equation}
            \begin{aligned}
                  \Gamma=\Bigg[\frac{2CE_fL_{max}^2}{1-\zeta^2} & +\frac{E_f^2L_{max}^2}{1-\zeta}\left(\frac{4\zeta}{1-\zeta^2}+\frac{E_f-1}{E_f}\right)\Bigg]\eta^2 \\
                                                                & +L_{max}(C+1)\eta-1
            \end{aligned}
      \end{equation}
      Particularly, if the learning rate $\eta$ is set properly to make $\Gamma\leq 0$, the terms related with $\Vert\boldsymbol{g}^t_k\Vert^2$ can be eliminated and the conclusion Eq.(\ref{thm}) can be obtained.
\end{proof}
Further, we summarize the following remarks to interpret some insights observed from \textbf{Theorem \ref{convergence thm}}.
\begin{remark}
      \label{gradient convergence}
      \textbf{Convergence of gradients.} \rm The theorem investigates the learning convergence from the perspective of gradients. When the average 2-norm of $\ell_{\mathcal{A}_k}(\overline{\boldsymbol\theta}_k)$'s gradients are upper bounded, one can deem that the loss functions are stable and the learning process converges.
\end{remark}
\begin{remark}
      \label{generalization}
      \textbf{Generalization to other metrics.} \rm Note that in the above theorem, we consider the actor loss without regard to the specific metric and the goal is to bound the deviation of the gradients. Therefore, if one hopes to maintain any metric (or metric combination) which converges to certain value, the theorem always holds.
\end{remark}
\begin{remark}
      \label{interpretation}
      \textbf{Interpretation of the bounds.} \rm In the theorem, the right-hand side upper bound $\Lambda\left(\omega,E_f,\boldsymbol{\mathcal{T}}_k^{T_0}\right)$ can be interpreted by the following three separated terms:
      \begin{equation}
            \label{interpret thm}
            \begin{aligned}
                  \Lambda\left(\omega,E_f,\boldsymbol{\mathcal{T}}_k^{T_0}\right) & =\underbrace{\frac{2\sum\limits_{k=1}^{N_e}\left[\ell_{\mathcal{A}_k}(\overline{\boldsymbol{\theta}}^{T_0}_k)-\ell_{\mathcal{A}_k}^*\right]}{\eta N_eT}}_{\rm Initial\ Deviation}+\underbrace{\frac{\eta\sum\limits_{k=1}^{N_e}L_{\boldsymbol{\mathcal{T}}_k}\sigma_{\boldsymbol{\mathcal{T}}_k}^2}{N_eB}}_{\rm Sequel\ Deviation} \\
                                                                                  & +\underbrace{\frac{\eta^2\sigma_{max}^2L_{max}^2}{2B}\Big(\frac{1+\zeta^2}{1-\zeta^2}E_f-1\Big)}_{\rm Federated\ Updating\ Deviation}.
            \end{aligned}
      \end{equation}
      The first term is the initial deviation caused by the training before $T_0$. The second term, related to the interaction after $T_0$, is called the sequel deviation. The third term is the gradients' deviation when the edge federated updating mode is carried.
\end{remark}
\begin{remark}
      \label{minimum bound}
      \textbf{Minimum bounds.} \rm For the federated update matrix in Eq.(\ref{Omega}), $\boldsymbol{\Omega}$ is positive-definite and symmetrical which has a 1-order eigenvalue $\Lambda_1=1$ and $N_e-1$ order repeated eigenvalue, $\Lambda_2=\cdots=\Lambda_{N_e-1}=\frac{N_e\omega-1}{N_e-1}$. Then, $\zeta=\frac{N_e\omega-1}{N_e-1}\in\left[\frac{-1}{N_e-1},1\right]$. Besides, one can rewrite the last term of the right-hand side in Eq.(\ref{thm}) as:
      \begin{equation}
            \label{mini bound}
            \begin{aligned}
                   & \frac{\eta^2\sigma_{max}^2L_{max}^2}{2B}\Big(\frac{1+\zeta^2}{1-\zeta^2}E_f-1\Big) \\
                   & =\frac{\eta^2\sigma_{max}^2L_{max}^2}{2B}\Big(\frac{2}{1-\zeta^2}E_f-E_f-1\Big).
            \end{aligned}
      \end{equation}
      Thus, with other hyper parameters fixed, the upper bound of the gradients' 2-norm in Eq.(\ref{thm}) meets its minimum if $\zeta^*=0$, i.e., $\omega^*=\frac{1}{N_e}$. This implies that when the edge-federated updating is carried out with uniform weights, namely, all elements in $\Omega$ are equal to $\frac{1}{N_e}$, the gradients’  mean-square time average is bounded tightly, which leads to best training convergence.
\end{remark}
\begin{remark}
      \label{omega=1}
      \textbf{Impact of $\omega$.} \rm Intuitively, if one sets $\omega<\frac{1}{N_e}$, the diagonal elements of $\Omega$ in Eq.(\ref{Omega}) are smaller than the others. We regard these cases as chaos because each agent keeps less knowledge of the policies learned from its own observations. Also, in particular, when $\omega=1$, $\Omega$ is an identity matrix which refers to the original mixed H-MAAC, where each agent learns the policies individually and no parameter sharing occurs. In this case, $\zeta=1$ and the right-hand side of Eq.(\ref{thm}) becomes infinity, which means that the gradients are unbounded, and therefore, the convergence of original H-MAAC is not guaranteed under the above deductions. Thus, the effective interval of $\omega$ lies on $[\frac{1}{N_e},1)$. Overall, we elucidate the impact of $\omega$ as the following Fig. \ref{explain omega}. When $\omega\rightarrow\frac{1}{N_e}$, the federated parameter sharing becomes uniform and the edge agents tend to be homogeneous. When $\omega\rightarrow 1$, the diagonal elements of $\Omega$ dominate and each edge agents keeps more individuality by preserving the most of its own policies.
      \begin{figure}[htbp]
            \centering
            \begin{minipage}{0.3\textwidth}
                  \includegraphics[width=1\textwidth]{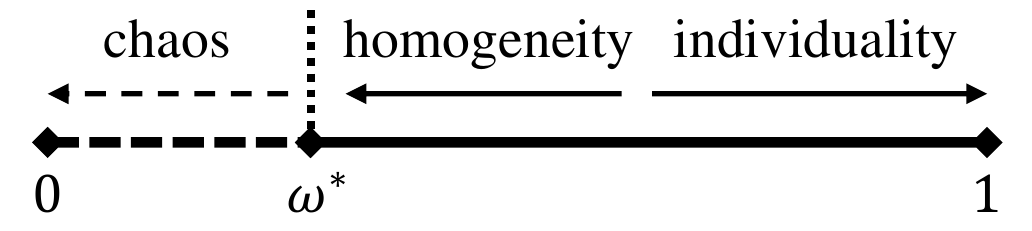}
            \end{minipage}%
            \caption{A sketch for the impact of $\omega$.}
            \label{explain omega}
      \end{figure}
\end{remark}
% \begin{remark}
%     \label{omega=1}
%     \textbf{Unbounded case.} Particularly, when $\omega=1$, $\Omega$ is an identity matrix which refers to the original H-MAAC (or mixed MAAC) where each agent learns the policies individually and no parameter sharing occurs. In this case, $\zeta=1$ and the right-hand side of Eq.(\ref{thm}) becomes infinity, which means that the gradients are unbounded and therefore, the convergence of original H-MAAC is not guaranteed under the above deductions.
% \end{remark}
\begin{remark}
      \label{efficiency convergence}
      \textbf{Convergence vs. efficiency.} \rm Note that in assumption 3) to 5), we require the preconditions that the interaction sets are given and fixed. Nevertheless, in the online MEC collaboration, the environment is stochastic where the dynamics of data arrival and $\epsilon$-exploration are time variant. In addition, the random sampling in experience replay may also lead to variant training sets, $\{\boldsymbol{\mathcal{T}}_k\}^{t}$. For reasons above, the optimal critic parameters $\boldsymbol{\phi}_k^*$, the Lipschitz constant $L_{\boldsymbol{\mathcal{T}}_k}$ in Eq.(\ref{l-smooth}), the gradient variance $\sigma_{\boldsymbol{\mathcal{T}}_k}$ as well as the coefficient $C$ in Eq.(\ref{sgd bound}) are likely to be unstable and of large variance at different learning epochs.
      Furthermore, to the contrary, the online collaboration system might benefit from the fluctuation of the gradients because the actor parameters $\boldsymbol{\theta}$ can be modified to adapt the change of the environment. Thus, the best training convergence in the ideal scenarios may not always guarantee the best system performance. We reckon this as a learning convergence vs. system efficiency trade-off in practice and the details will be discussed in Section \ref{section sim}.
\end{remark}

\section{Performance Evaluation}
\label{section sim}
In this section, we will test the proposed MEC collaboration framework and evaluate its performance through simulation. Besides, the comparisons with popular RL approaches and other insights are investigated.
\subsection{Simulation Settings}
Firstly, we implement the age sensitive MEC system presented in Section \ref{section system} as a universal \textit{gym} \cite{brockman2016openai} module in Python. For simplicity of the following discussion and simulation, we suppose that the packets independently arrive at each data source $S_n$ with the probability $p^n_g$ at the beginning of each time slot and the data size follows a Poisson distributed with the arrival rate, $\lambda_n$ \cite{fan2018application}. Thus, the data generation process can be described as a switch Poisson process:
\begin{equation}
    \label{data gen}
    P\left\{d_{n,i}(t)=m\right\}=\mathbbm{1}_{G_n}\cdot\frac{e^{-\lambda_n}\lambda_n^m}{m!}
\end{equation}
where the boolean random variable $G_n\sim{\rm Bernoulli}(p^n_g)$, indicates the arrival of the packets.

\begin{table*}[htbp]
    \caption{\upshape Main parameter settings for simulations.}
    \label{sim_settings}
    \centering
    \begin{tabular}[c]{ccc}
        \hline
        \rowcolor{gray!20}Parameter                  & Description                                                    & Value                               \\
        \hline
        $(\lambda_n,p_g^n)$                          & Data rate and arrival probability of data generation at $S_n$. & (1Kb/slot, 0.3)                     \\
        $(r^k_{move}$, $r^k_{obs}$, $r^k_{collect})$ & Edge devices' radius of movement, observation, collection.     & (6, 60, 40)                         \\
        $B_{col}^k$, $B_{exe}^k$                     & Maximum buffer size for collected data and executed data.      & 5                                   \\
        $f_c^k$                                      & Computation rate of of$E_k$.                                   & 20Kb/slot                           \\
        $W$                                          & Total bandwidth for offloading communication.                  & 100MHz                              \\
        $f$                                          & The carrier frequency in Eq.(\ref{PL}).                        & 2.5GHz                              \\
        $(a,b,\eta_{LoS},\eta_{NLoS})$               & Coefficients of A2G path loss.                                 & (9.61, 0.16, 1, 20)                 \\
        $N_0$                                        & The noise power spectral density of A2G channel.               & -130dB                              \\
        $P^k_{tr,max}$                               & Maximum power for offloading communication.                    & 0.2W                                \\
        $\gamma$                                     & Penalty decay.                                                 & 0.85                                \\
        $\tau$                                       & Target updating weight.                                        & 0.8                                 \\
        $\epsilon$                                   & Probability of random exploration.                             & 0.2                                 \\
        $T_u,E_f$                                    & Period for target updating and federated updating.             & 8                                   \\
        $(\eta_{\mathcal{A}}$, $\eta_{\mathcal{C}})$ & Learning rates for actor/critic nets.                          & $(1\times 10^{-3},2\times 10^{-3})$ \\
        $B$                                          & Batch size of experience replay.                               & 128                                 \\
        \hline
    \end{tabular}
\end{table*}

Secondly, let us claim some basic environment settings of simulations. The main parameter settings are listed in Table \ref{sim_settings}. As for data sources, we set the arrival probability and the generating rate as $0.3$ and $1$Kb/slot. Then, for the attributes of edge devices, we set $r^k_{move}$, $r^k_{obs}$, $r^k_{collect}$ to be $6$, $60$, $40$ in measure of the grid map and the height is fixed. The computing rate of edge process is $20$Kb/slot and the maximum buffer lengths for caching collected data and executed data are both set to be $5$ packet pieces. For edge-cloud communication channel, we set the total offloading bandwidth as $100$MHz, the carrier frequency $f$ in Eq.(\ref{PL}) as $2.5$GHz, the noise power spectral density $N_0$ as $-130$dB, and the max transmission power $P^k_{tr,max}$ as $0.2$W. The coefficients $(a,b,\eta_{LoS},\eta_{NLoS})$ in Eq.(\ref{PL}) and Eq.(\ref{p_los}) are selected to be $(9.61, 0.16, 1, 20)$, which refers to the urban scenarios mentioned in \cite{al2014optimal}. Particularly, the transmission rate of edge-source collection is fixed to be $8$Kb/slot due to the limit of collection cover.

With regard to learning configurations, the decay coefficient $\gamma$ of the system penalty is set to be 0.85. As for hyper parameters, the learning rates of actor and critic are 1e-3, 2e-3, respectively and the batch size of each epoch is 128. As proposed in Section \ref{section alg}, the target updating period $T_u$ and the reserving weight $\tau$ are 8 and $0.8$. In addition, we exploit the $\epsilon$-exploration with $\epsilon=0.2$ and edge-federated mode where parameters are shared every 8 learning epochs.

As for the metrics, we also adopt PAoI (peak AoI) and worst AoI in comparisons. PAoI of data source $S_n$, $\Delta_{p,n}$, is defined as the average peak value of $S_n$'s AoIs, which represents the maximum age of information before a new update is received \cite{costa2014age}. We denote
\begin{equation}
    \label{paoi}
    \overline{\Delta}_p=\frac{1}{N_s}\sum\limits_n\Delta_{p,n}
\end{equation}
as the average PAoI for all data sources. Besides, the worst AoI, defined as the maximum AoI of data sources at each time slot, is considered to evaluate the AoI performance in worst cases.

\subsection{Evaluation Results}
\begin{figure*}
    \centering
    \subfigure[The evolution of MEC system's average age, $\overline{\Delta}(t)$.]{
        \label{age compare} %% label for first subfigure
        \begin{minipage}{0.38\textwidth}
            \includegraphics[width=1\textwidth]{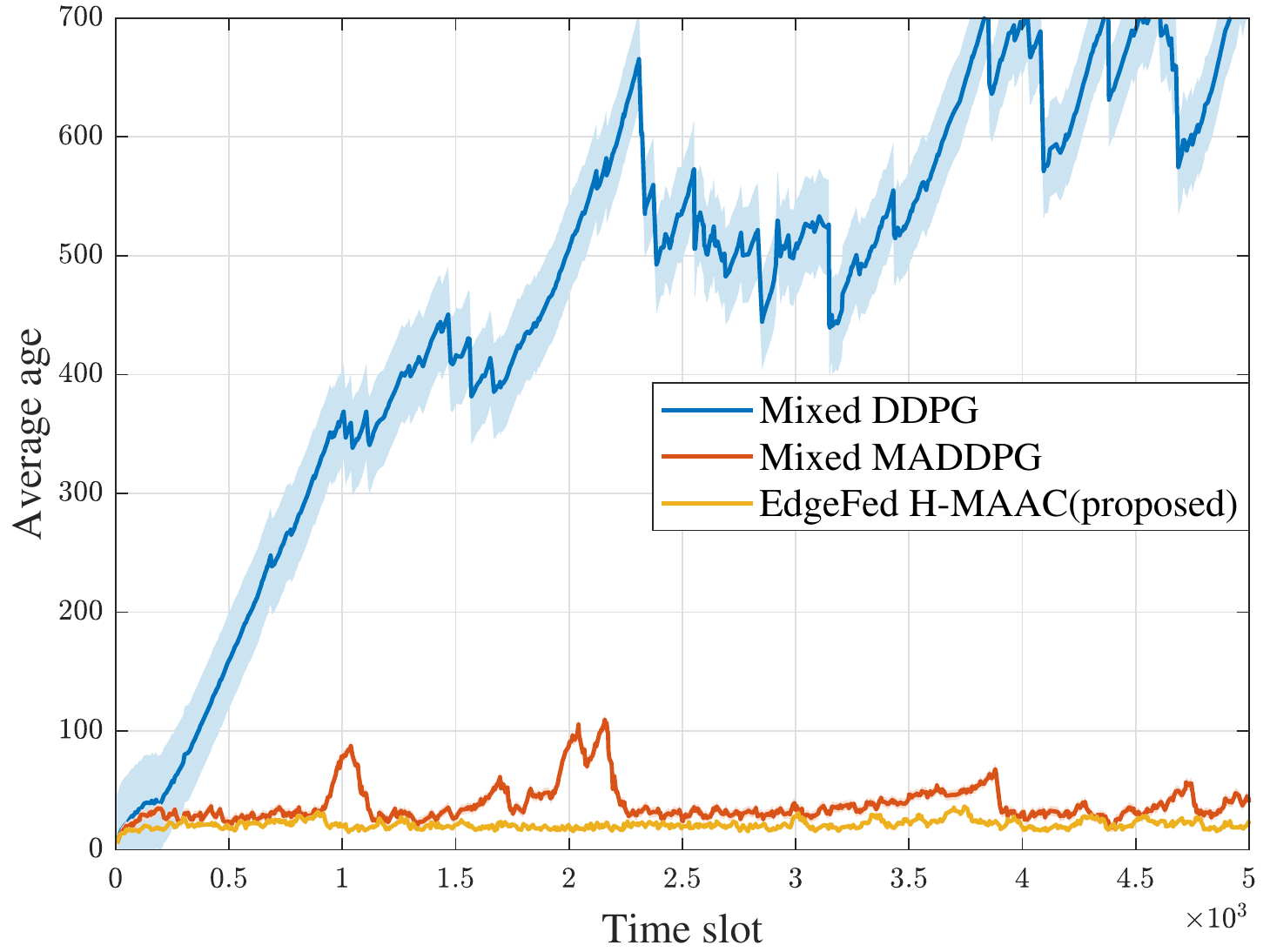}
        \end{minipage}}%
    \subfigure[The worst AoI of all data sources, i.e., $\max \left\{\Delta_k(t)\right\}$.]{
        \label{worst compare} %% label for second subfigure
        \begin{minipage}{0.38\textwidth}
            \includegraphics[width=1\textwidth]{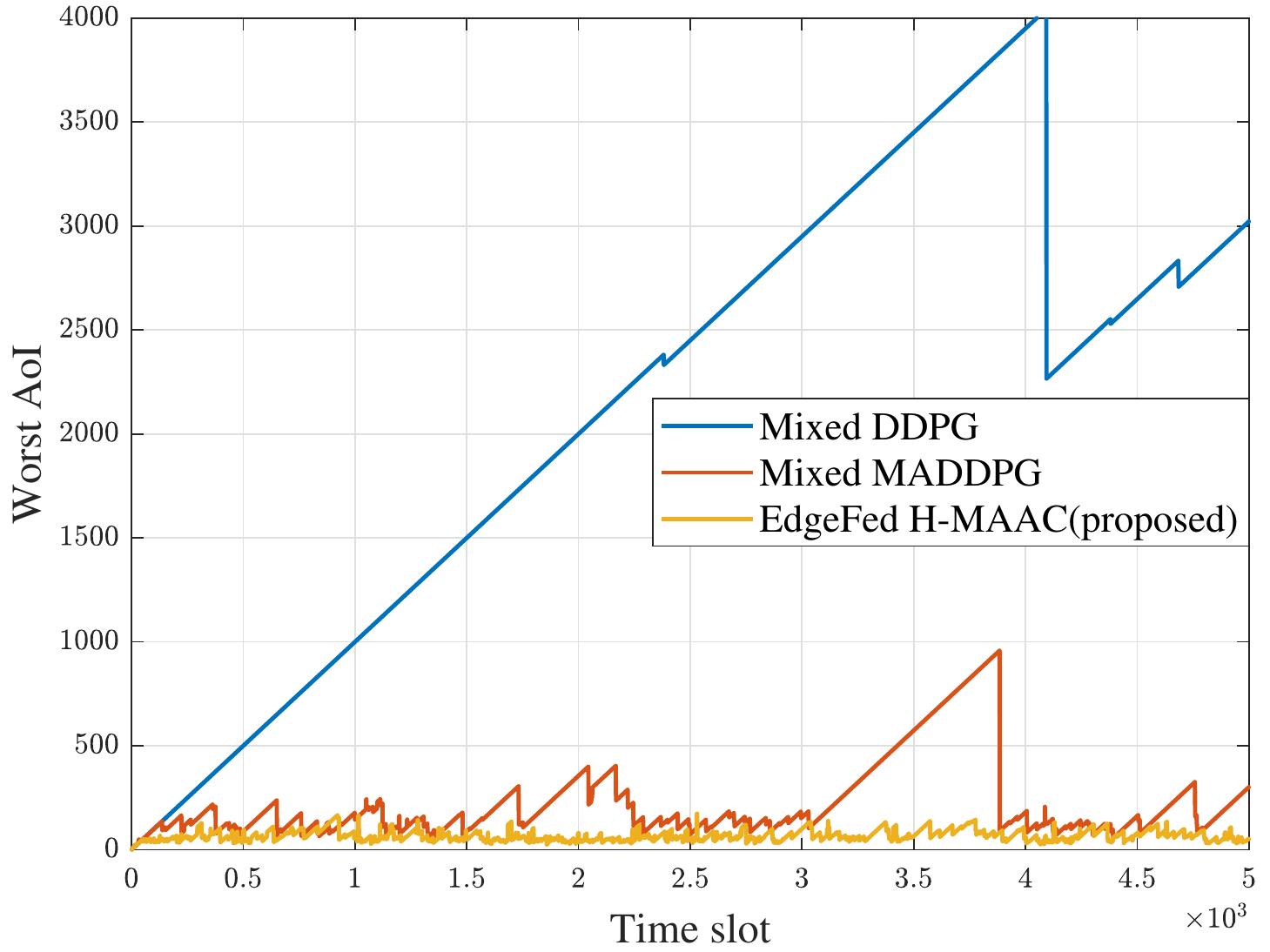}
        \end{minipage}}

    \subfigure[The data volume received at cloud center.]{
        \label{data compare} %% label for second subfigure
        \begin{minipage}{0.38\textwidth}
            \includegraphics[width=1\textwidth]{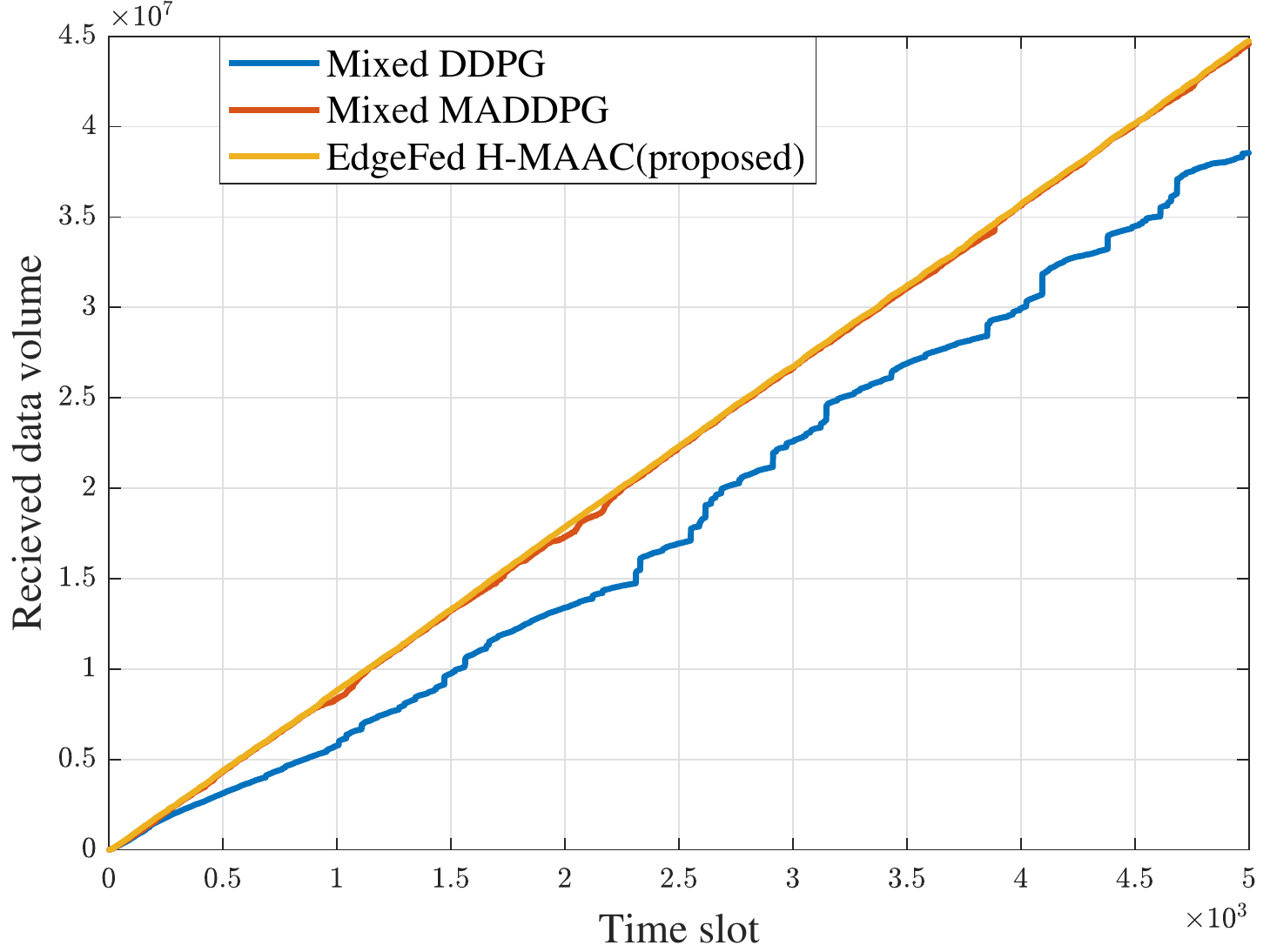}
        \end{minipage}}
    \subfigure[The count of the data packets received by cloud center.]{
        \label{packet compare} %% label for second subfigure
        \begin{minipage}{0.38\textwidth}
            \includegraphics[width=1\textwidth]{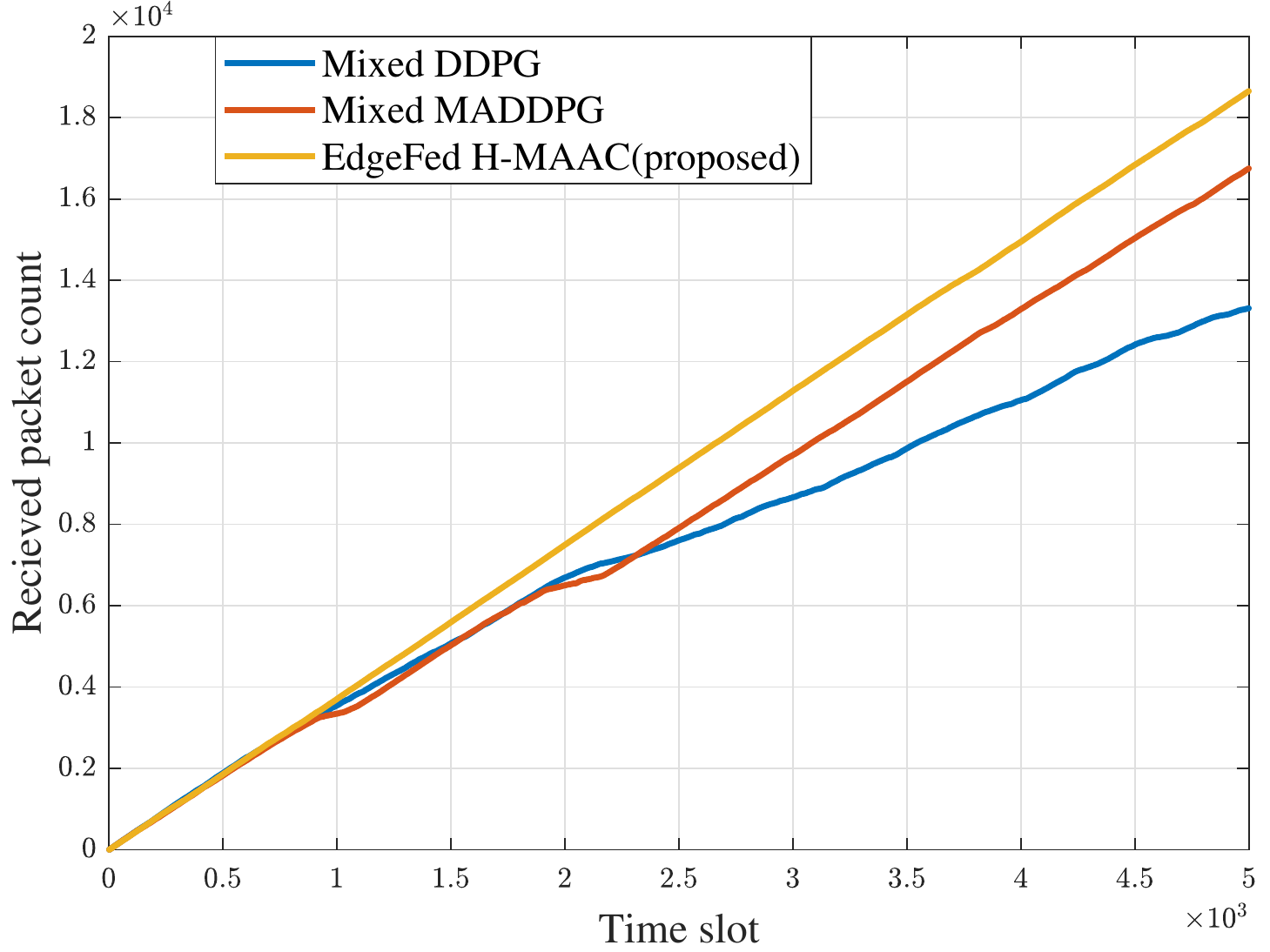}
        \end{minipage}}
    \caption{The comparison of several RL based MEC collaboration methods. (4 edge servers, 30 data sources on a $200\times 200$ map.)}
    \label{compare} %% label for entire figure
\end{figure*}

\begin{table}[htbp]
    \caption{\upshape A numerical comparison on AoI metrics.}
    \label{compare_tab}
    \centering
    \begin{tabular}{c|c|c|c|c|c}
        \hline

        \multirow{3}*{Approach}        & \multicolumn{2}{c|}{Average age $\overline{\Delta}(t)$} & \multicolumn{3}{c}{Peak AoI}                                                                              \\
        \cline{2-6}
                                       & \multirow{2}*{mean}                                     & \multirow{2}*{std}           & \multirowcell{2}{peak                                                      \\ count} & \multirowcell{2}{$\overline{\Delta}_p$}&\multirowcell{2}{$var(\Delta_p)$}   \\
                                       &                                                         &                              &                         &                        &                         \\
        \hline
        \multirowcell{2}{Mixed DDPG}   & \multirowcell{2}{473.3}                                 & \multirowcell{2}{174.4}      & \multirowcell{2}{10121} & \multirowcell{2}{26.6} & \multirowcell{2}{156.0} \\&&&&&\\
        \multirowcell{2}{Mixed MADDPG} & \multirowcell{2}{39.3}                                  & \multirowcell{2}{16.7}       & \multirowcell{2}{12842} & \multirowcell{2}{24.9} & \multirowcell{2}{93.0}  \\&&&&&\\
        \multirowcell{2}{EdgeFed H-MAAC                                                                                                                                                                      \\(proposed)} & \multirowcell{2}{\textbf{21.2}}  & \multirowcell{2}{\textbf{3.6}} & \multirowcell{2}{\textbf{13782}} & \multirowcell{2}{\textbf{21.7}} & \multirowcell{2}{\textbf{23.1}} \\&&&&&\\
        \hline
    \end{tabular}
\end{table}

We implement EdgeFed H-MAAC collaboration algorithm, by building the CNN-MLPs mixed networks for edge actor-critic nets and MLP based networks as center agent with TensorFlow. To compare the performance of the proposed frameworks with the other RL approaches, we select two actor-critic based algorithms, DDPG (centralized) and the popular MADDPG (multi-agent) as the baselines where dual neural networks are also exploited to learn the mixed policies. For fairness, we set the same random seeds of the MEC environment for all methods and instead of spending much effort on network tuning, we also fix the random seeds of training processes. Consequently, the results are of generality and can be easily reproduced.

Fig. \ref{compare} and Table \ref{compare_tab} show the comparison results of proposed algorithm (set $\omega=0.5$) and the baselines under the environment with 4 edge devices and 30 data sources on a $200\times 200$ map where the positions of edge devices and data sources are initialized randomly. The average age $\overline{\Delta}(t)$ during the online interaction is shown in Fig. \ref{age compare} where one can find that the mixed DDPG results in highest $\overline{\Delta}(t)$ and the curve is not stable till 5K epoch. However, under mixed MADDPG and EdgeFed H-MAAC, the average age maintains at a lower value after sufficient iterations. Specifically, the statistics of $\overline{\Delta}(t)$ are listed in Table \ref{compare_tab}. EdgeFed H-MAAC attains not only the lowest average age, but also the lowest variance, which means the edge-federated approach outperforms DDPG as well as MADDPG on both system penalty and learning stability. The right-hand part of Table \ref{compare_tab} presents the comparison on PAoI. Evidently, the proposed collaboration algorithm reaps most peak updates and lowest average PAoI, $\overline{\Delta}_p$. This implies that the efficiency of data processing in MEC can be promoted by EdgeFed H-MAAC. The lowest variance of PAoI also demonstrate that all data sources are updated frequently and fairly. Fig. \ref{worst compare} displays the worst AoI of three approaches. EdgeFed H-MAAC also performs the best. One can find that under the centralized DDPG, some sources are ignored for a long time. We explain this as the fact that the centralized collaboration algorithms require larger neural network models with complex structure to extract the relations between the excessive global input states and the local policies of each individual agent, which also leads to the difficulties for training. Besides, Fig. \ref{data compare} and \ref{packet compare} present the volume and the count of the aggregated packets received at the cloud center, namely, the final hop of the MEC system. The curves illustrate that EdgeFed H-MAAC collaboration algorithm also improves the data utility of the MEC system by finishing more data processing within same time.

Additionally, we present the training loss of mixed MADDPG and EdgeFed H-MAAC in Fig. \ref{loss fig} where the actor loss and critic loss of the first edge agent and the center agent are displayed. Similarly, EdgeFed  H-MAAC leads to lower and more stable loss. What's more, the critic loss shown in Fig. \ref{ecloss} and Fig. \ref{ccloss} demonstrate that it is acceptable to assume that the center agent training can be out of consideration and the critic nets fit the $Q$ values well, which refers to the assumption 2) and 3) of the convergence discussion in Section \ref{section convergence}. Beyond our expectation, although the federated parameter sharing only works on edge agents, this scheme also promotes the center agent training significantly.
\begin{figure}[htbp]
    \centering
    \subfigure[Edge Agent Actor Loss.]{
        \label{ealoss} %% label for first subfigure
        \begin{minipage}{0.48\linewidth}
            \includegraphics[width=1\textwidth]{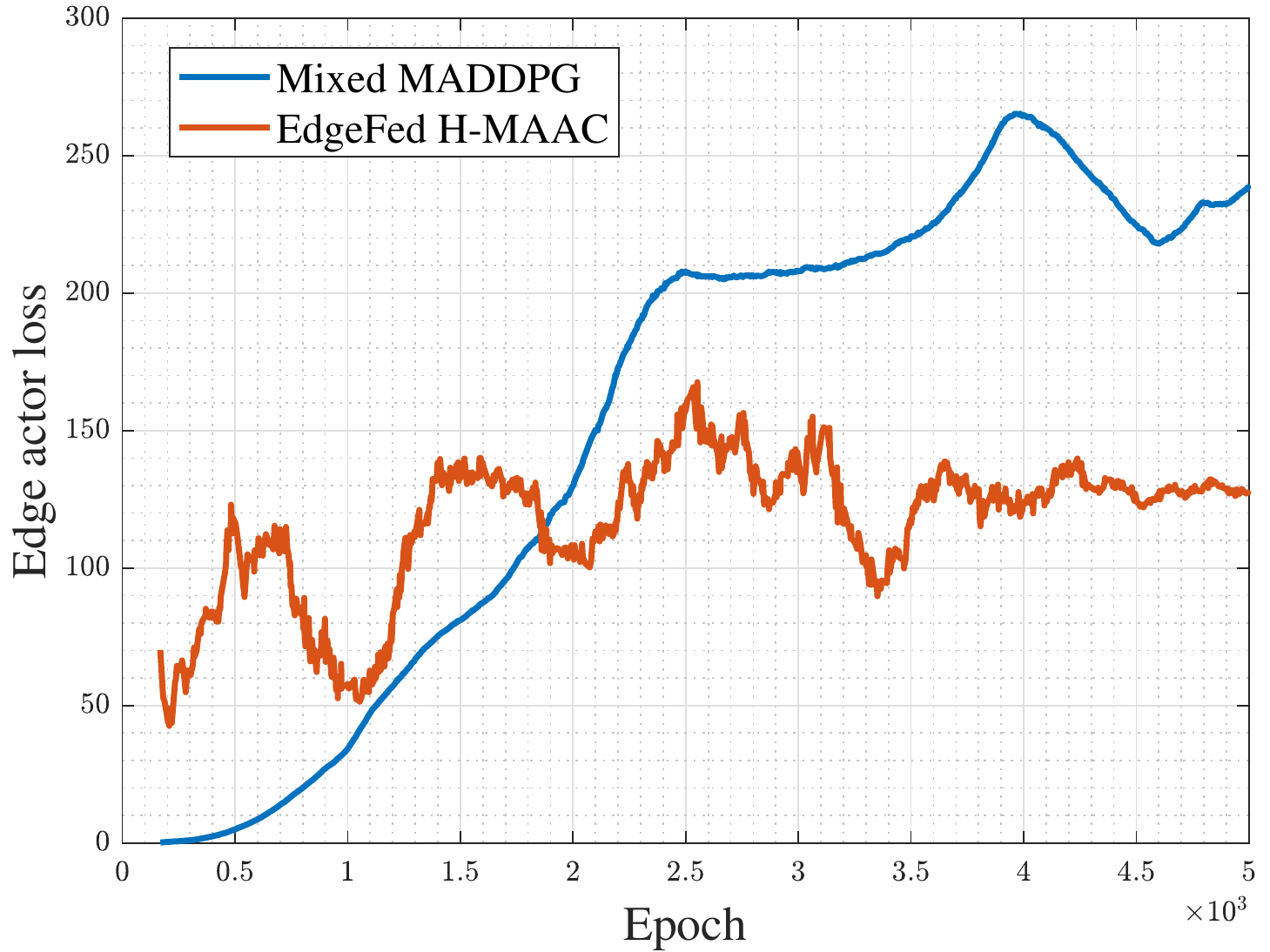}
        \end{minipage}}%
    \subfigure[Edge Agent Critic Loss.]{
        \label{ecloss} %% label for second subfigure
        \begin{minipage}{0.48\linewidth}
            \includegraphics[width=1\textwidth]{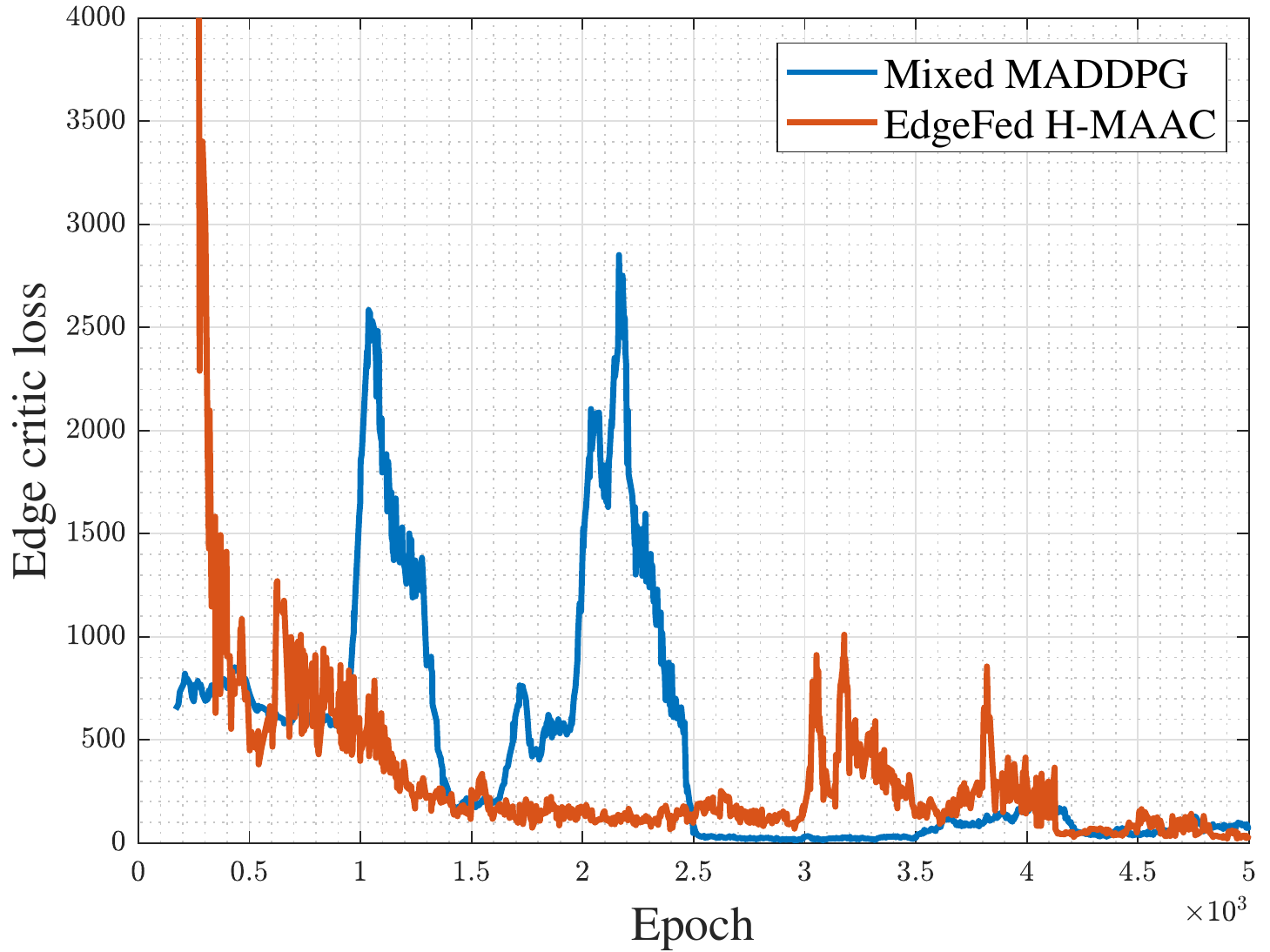}
        \end{minipage}}

    \subfigure[Center Agent Actor Loss.]{
        \label{caloss} %% label for first subfigure
        \begin{minipage}{0.48\linewidth}
            \includegraphics[width=1\textwidth]{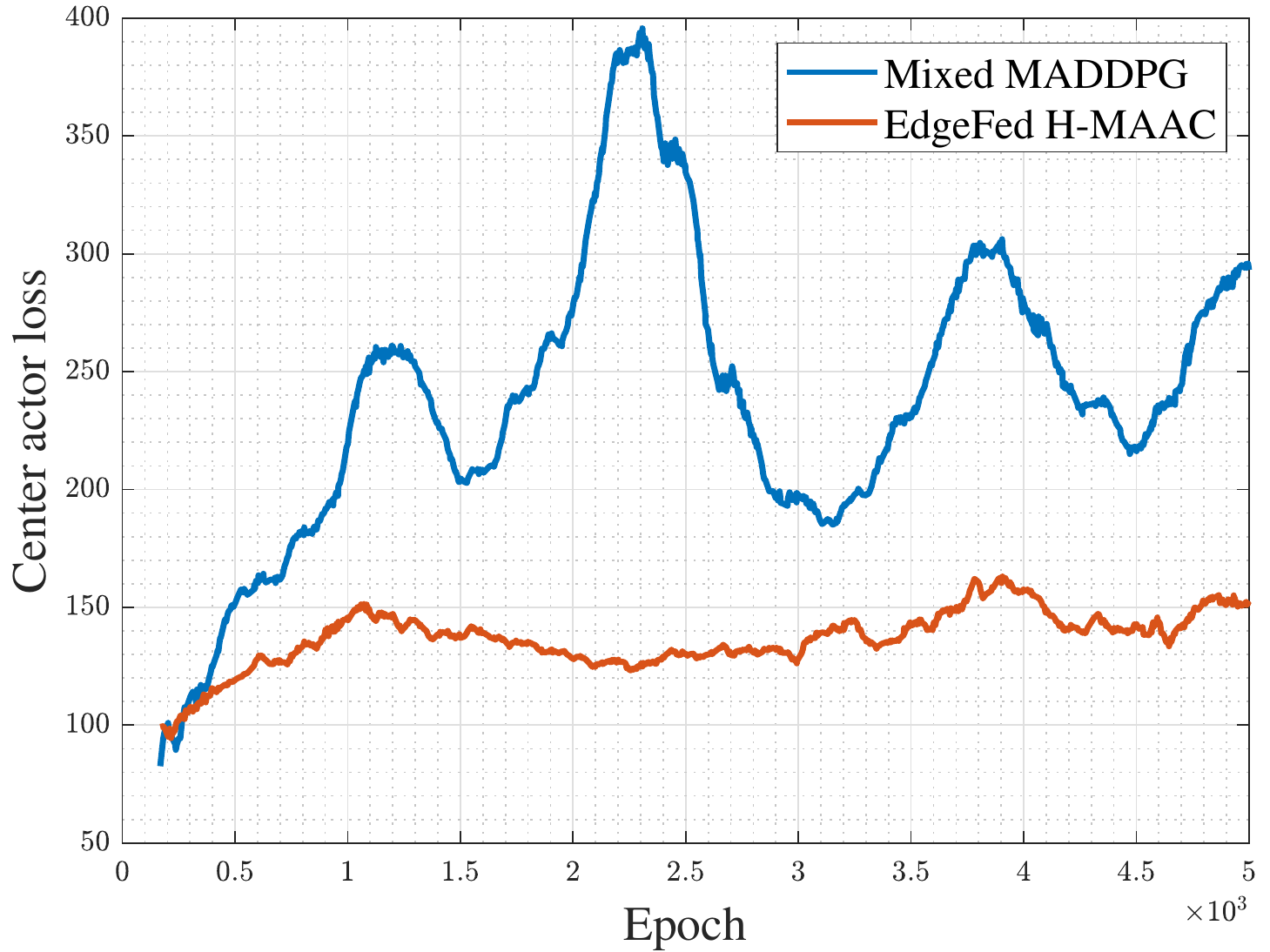}
        \end{minipage}}%
    \subfigure[Center Agent Critic Loss.]{
        \label{ccloss} %% label for second subfigure
        \begin{minipage}{0.48\linewidth}
            \includegraphics[width=1\textwidth]{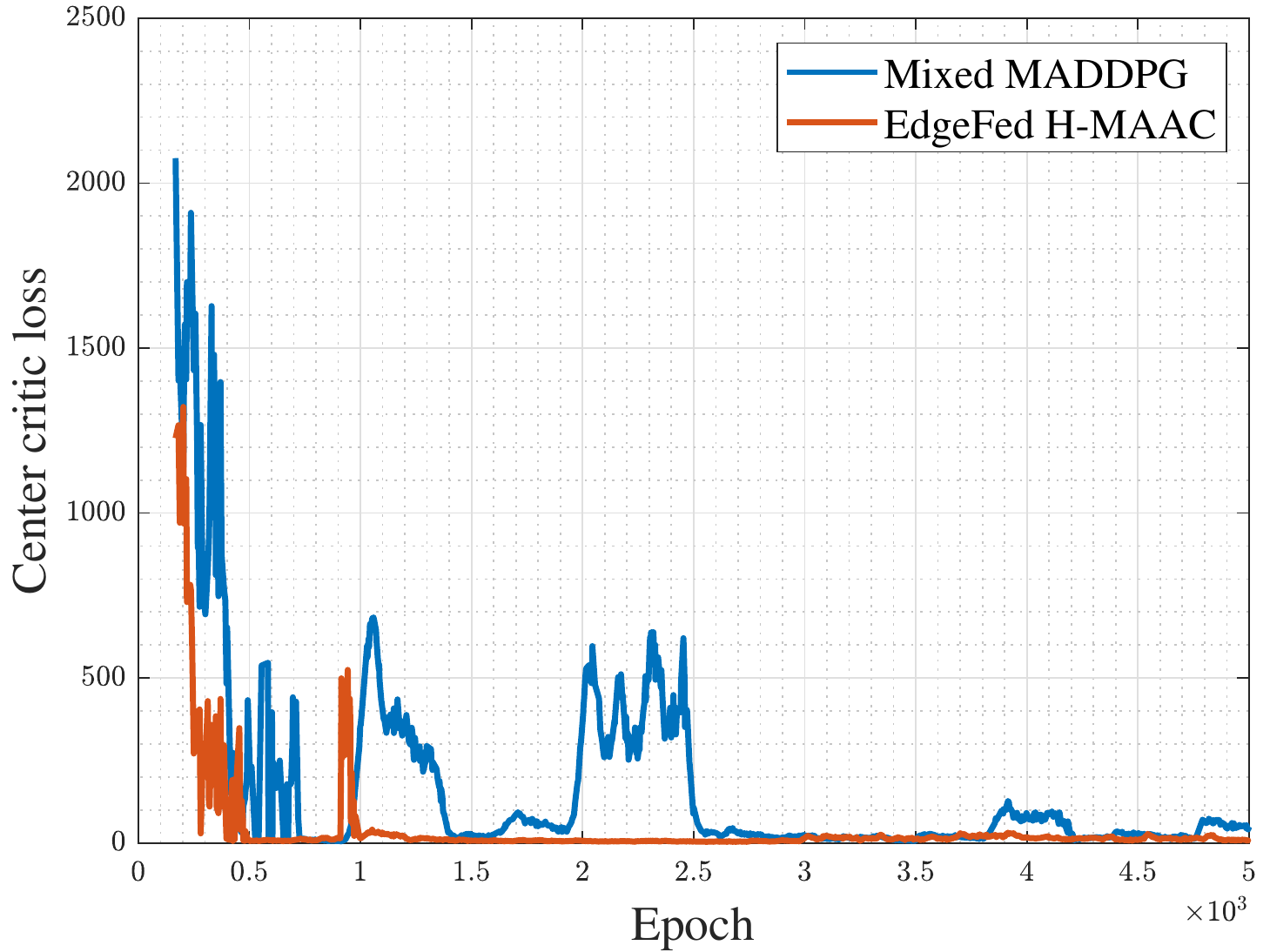}
        \end{minipage}}
    \caption{Training loss of all agents' actor-critic nets.}
    \label{loss fig} %% label for entire figure
\end{figure}

\begin{figure}[htbp]
    \centering
    \begin{minipage}{0.4\textwidth}
        \includegraphics[width=1\textwidth]{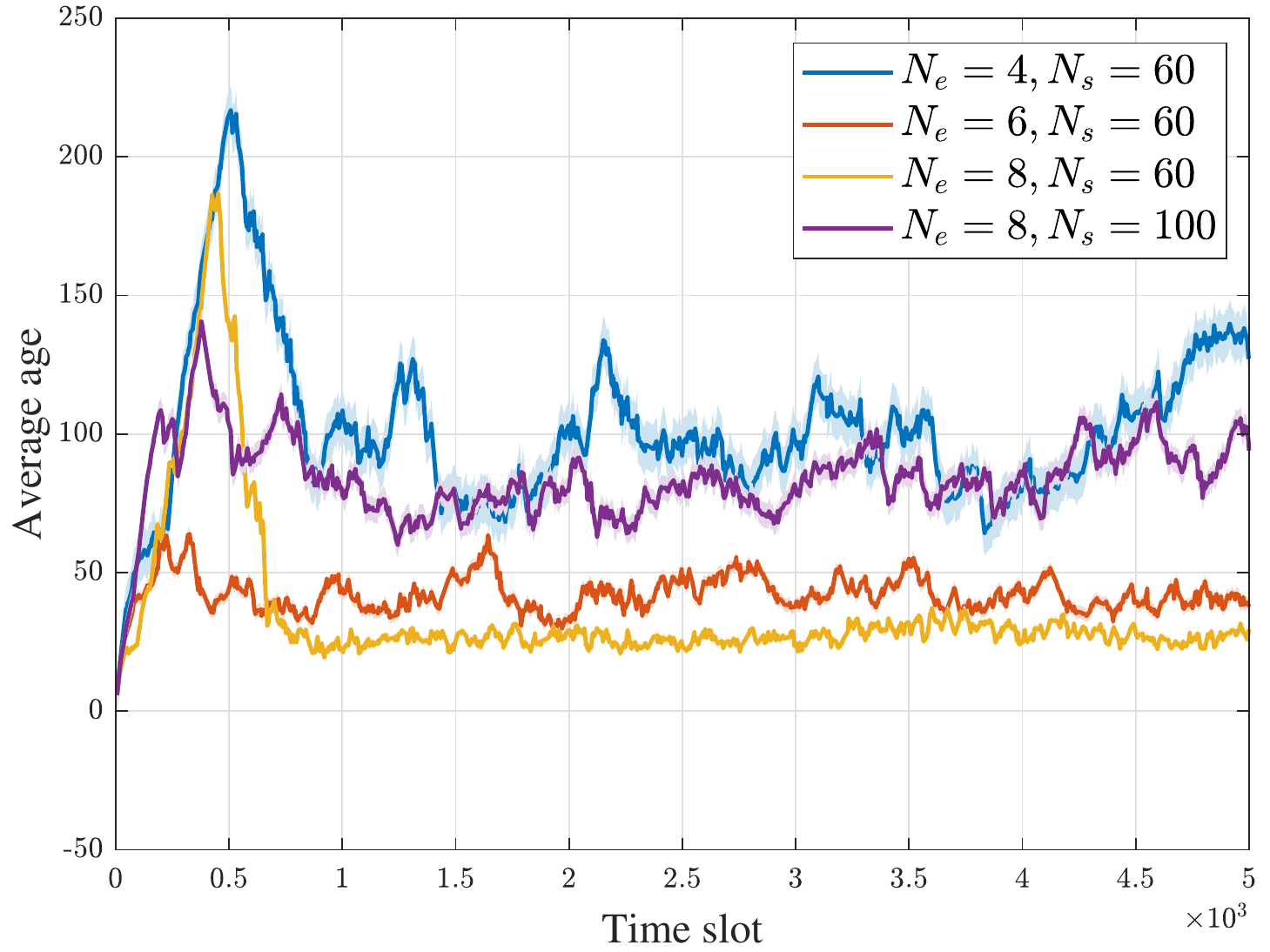}
    \end{minipage}%
    \caption{EdgeFed H-MAAC's $\overline{\Delta}(t)$ performance under different environment settings on a $300\times 300$ map.}
    \label{param figs}
\end{figure}

\begin{figure}[htbp]
    \centering
    \subfigure[$\overline{\Delta}(t)$ performance under different $\omega$. The vertical dot lines represent the convergence time of each $\omega$.]{
        \begin{minipage}{0.47\textwidth}
            \includegraphics[width=1\textwidth]{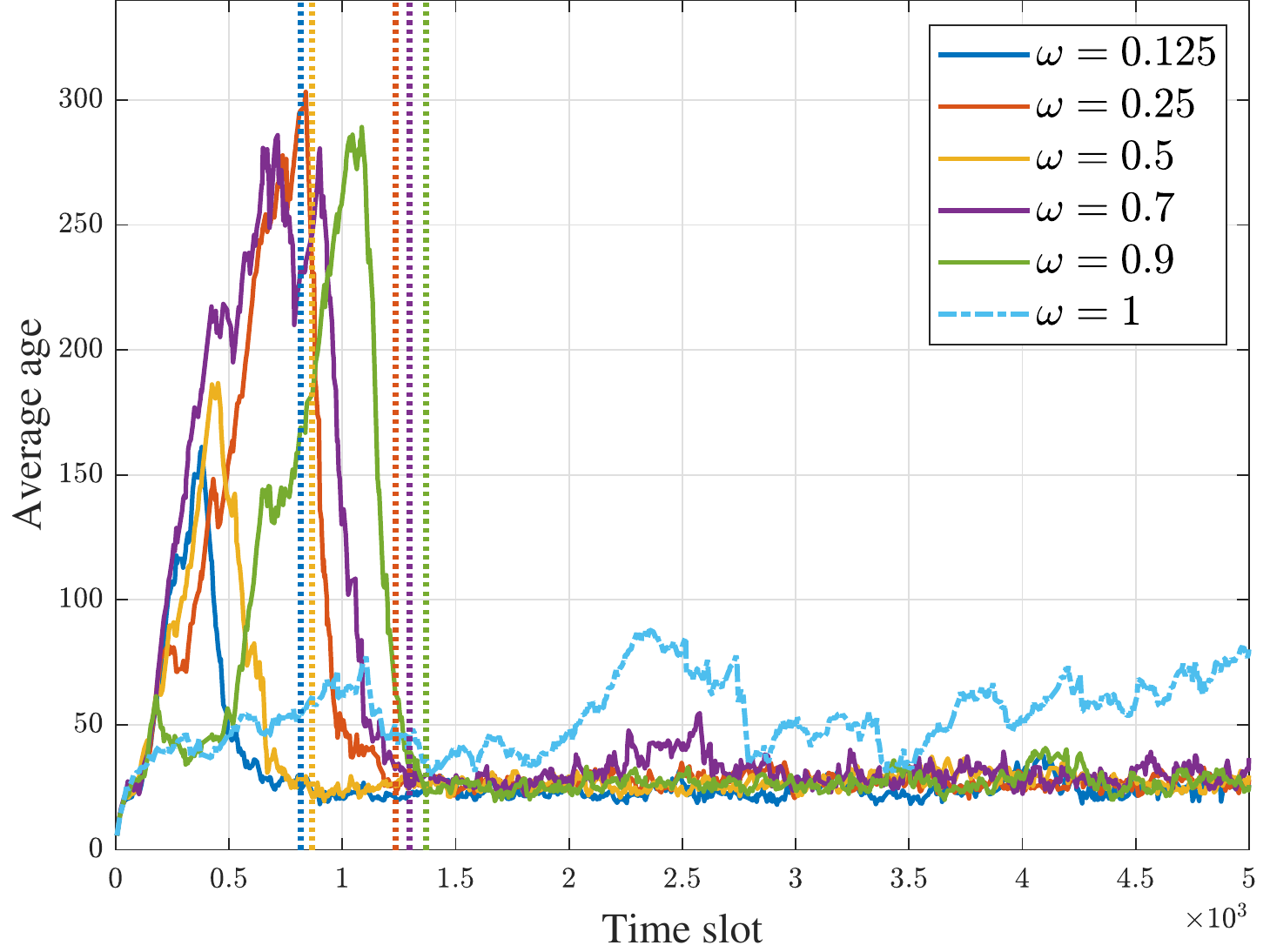}
        \end{minipage}%
        \label{omega curve}}
    \subfigure[Box plots of $\overline{\Delta}(t)$ under different $\omega$.]{
        \begin{minipage}{0.47\textwidth}
            \includegraphics[width=1\textwidth]{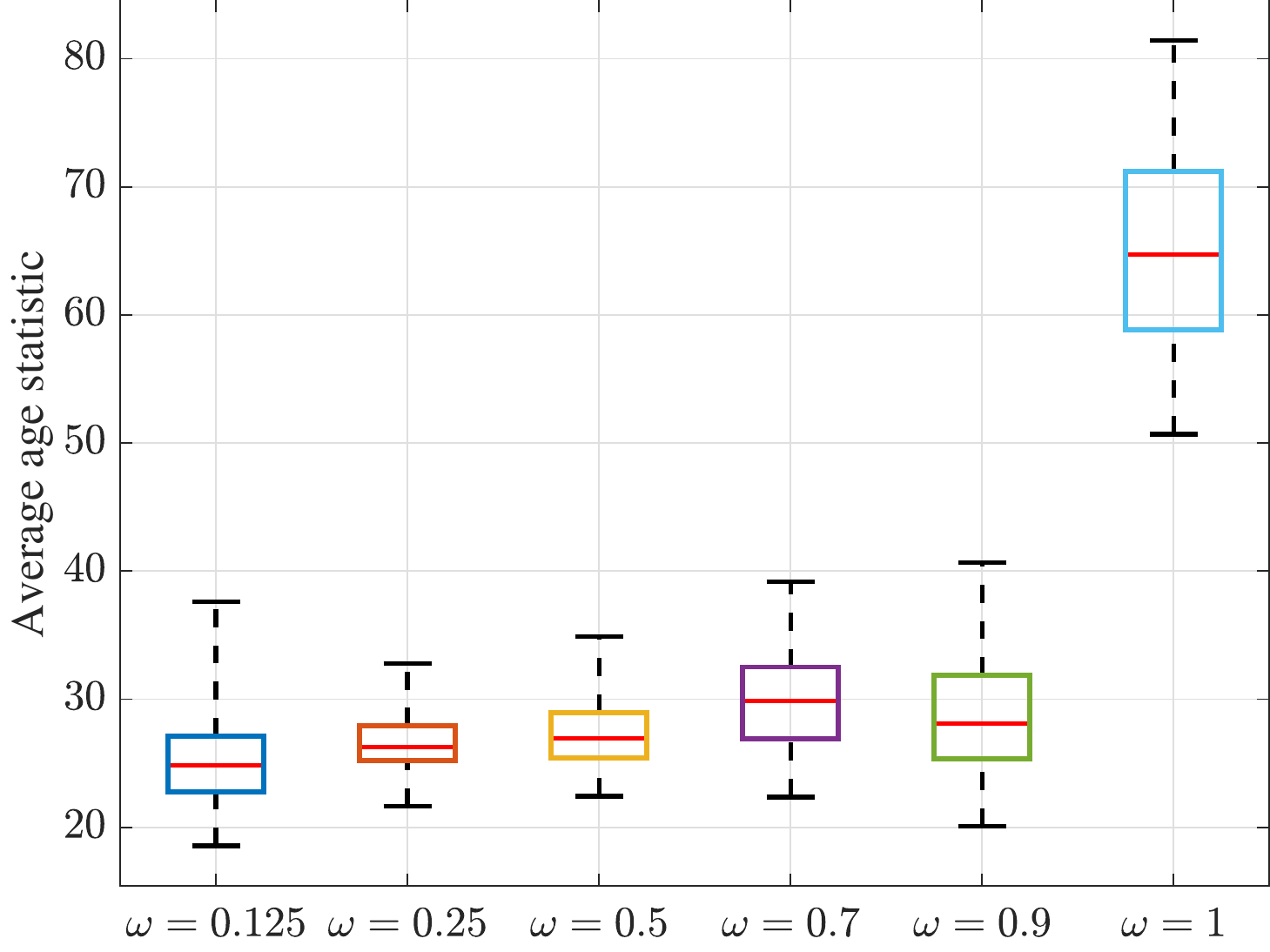}
        \end{minipage}%
        \label{omega box}}
    \caption{EdgeFed H-MAAC performance with different $\omega$. (8 edge devices, 60 data sources on a $300\times 300$ map.)}
    \label{omega figs}
\end{figure}
Then, we investigate the performance under different environment settings. We change the edge number $N_e$ as well as the source number $N_s$ and evaluate the EdgeFed H-MAAC in these cases. As presented in Fig. \ref{param figs}, low average ages are well maintained under different environment settings through EdgeFed H-MAAC collaboration. Consistent with the common sense, the more edge servers or fewer data sources both lead to better timeliness of the MEC systems.

Furthermore, to investigate the impact of the federated factor $\omega$, we set different $\omega$ in the scene with 8 edge devices and 60 data sources on a $300\times 300$ map. Likewise, the settings and the randomness of the MEC environment are identical for all the simulations. Note that the federated factor $\omega$ denote the weight with which each agent retains its model during the parameter sharing, agents will lose their own parameters if $\omega$ becomes too small. Hence, we only explore the cases where $\omega\geq\frac{1}{N_e}$. The results are presented as $\overline{\Delta}(t)$ evolution curves in Fig. \ref{omega curve} and box plots of $\overline{\Delta}(t)$ in Fig. \ref{omega box}. Intuitively, one can find that while $\omega=1$, i.e., the original H-MAAC, the system gets the worst performance and the learning process seems not to converge after 5K iterations, which is consistent with our theoretical analysis in \textbf{Remark \ref{omega=1}}. Additionally, from \textbf{Theorem \ref{convergence thm}} and \textbf{Remark \ref{minimum bound}}, the deviation of the gradients gets minimum when $\omega^*=0.125$. This is verified by the experiment where the vertical dot lines imply that the average age gets the rapidest convergence when $\omega=\omega^*$ and the larger $\omega$ leads to longer convergence time. However, the simulation outcomes show that $\omega=0.25$ and $\omega=0.5$ also perform well as they reach low average ages and even lower variances with similarly rapid convergence. As discussed in \textbf{Remark \ref{efficiency convergence}}, due to the dynamics of the environment and the randomness of the online learning, the constants in Eq.(\ref{l-smooth}) and Eq.(\ref{sgd bound}) are time variant. Besides, the fine-fitness of edge critic nets are not exactly guaranteed, as shown in Fig. \ref{ecloss} where the edge critic loss of EdgeFed H-MAAC does not strictly decrease to 0. For above reasons, one can find the gap between the convergence theorem and the simulation results. We explain such phenomenon as a trade-off between the ideal learning convergence and the system robustness to environment variation. In practice, the fluctuation of the gradients may contribute to learning the features of the stochastic environment. Meanwhile, smaller $\omega$ implies that the model parameters of each edge agent itself are preserved with lower proportion. Particularly, if $\omega=\omega^*$, all edge agents learn the same parameters after federated updating operation. Therefore, all edge agents tend to make same responses to the input states and lose their individuality small $\omega$. For larger federated factor $\omega$, though the gradients are not bounded tightly, it reserves the individuality of each edge agents to counter the stochastic environment in MEC collaboration systems. Thus, as elaborated in Fig. \ref{explain omega}, in such multi-agent cooperative learning framework, $\omega$ with mediate values may balance the learning convergence and the system efficiency. In addition, the simulation results in Fig. \ref{omega figs} can be utilized to design the proper edge-federated learning mode for the proposed multi-agent collaboration MEC algorithm. Approximately, through the above results and the discussions, the recommended interval of the federated factor $\omega$ lies on $[\frac{1}{N_e},0.5]$.

\section{Conclusion}
\label{section conclusion}
We investigated the age sensitive MEC systems and proposed a policy based multi-agent reinforcement learning framework, H-MAAC, for agent intelligent control of the trajectory planning, data scheduling as well as bandwidth allocation. By adopting federated learning mode, we developed the corresponding edge federated online joint collaboration algorithms whose convergence were theoretically proved. We implemented the MEC simulation system and evaluated the proposed algorithms. The outcomes showed that our method has lower average age and better learning stability compared to classical centralized actor-critic RL approaches. Moreover, some other advantages and the inspirations for edge federated design are also discussed according to the simulation results.

For further work, based on the proposed H-MAAC framework, more operations of agents can be expanded such as the power allocation, multitask scheduling and multitask offloading. In addition, the penalty/reward of the system can be flexibly defined which may bring more applications for cooperative MEC systems.

\nocite{*}
\bibliographystyle{IEEEtran}
\bibliography{IEEEabrv, refs}

\begin{thebibliography}{10}
\providecommand{\url}[1]{#1}
\csname url@samestyle\endcsname
\providecommand{\newblock}{\relax}
\providecommand{\bibinfo}[2]{#2}
\providecommand{\BIBentrySTDinterwordspacing}{\spaceskip=0pt\relax}
\providecommand{\BIBentryALTinterwordstretchfactor}{4}
\providecommand{\BIBentryALTinterwordspacing}{\spaceskip=\fontdimen2\font plus
\BIBentryALTinterwordstretchfactor\fontdimen3\font minus
  \fontdimen4\font\relax}
\providecommand{\BIBforeignlanguage}[2]{{%
\expandafter\ifx\csname l@#1\endcsname\relax
\typeout{** WARNING: IEEEtran.bst: No hyphenation pattern has been}%
\typeout{** loaded for the language `#1'. Using the pattern for}%
\typeout{** the default language instead.}%
\else
\language=\csname l@#1\endcsname
\fi
#2}}
\providecommand{\BIBdecl}{\relax}
\BIBdecl

\bibitem{shi2016edge}
W.~Shi, J.~Cao, Q.~Zhang, Y.~Li, and L.~Xu, ``Edge computing: Vision and
  challenges,'' \emph{IEEE internet of things journal}, vol.~3, no.~5, pp.
  637--646, 2016.

\bibitem{sun2016edgeiot}
X.~Sun and N.~Ansari, ``Edgeiot: Mobile edge computing for the internet of
  things,'' \emph{IEEE Communications Magazine}, vol.~54, no.~12, pp. 22--29,
  2016.

\bibitem{hu2015mobile}
Y.~C. Hu, M.~Patel, D.~Sabella, N.~Sprecher, and V.~Young, ``Mobile edge
  computing—a key technology towards 5g,'' \emph{ETSI white paper}, vol.~11,
  no.~11, pp. 1--16, 2015.

\bibitem{mao2017survey}
Y.~Mao, C.~You, J.~Zhang, K.~Huang, and K.~B. Letaief, ``A survey on mobile
  edge computing: The communication perspective,'' \emph{IEEE Communications
  Surveys \& Tutorials}, vol.~19, no.~4, pp. 2322--2358, 2017.

\bibitem{abd2019role}
M.~A. Abd-Elmagid, N.~Pappas, and H.~S. Dhillon, ``On the role of age of
  information in the internet of things,'' \emph{IEEE Communications Magazine},
  vol.~57, no.~12, pp. 72--77, 2019.

\bibitem{abbas2017mobile}
N.~Abbas, Y.~Zhang, A.~Taherkordi, and T.~Skeie, ``Mobile edge computing: A
  survey,'' \emph{IEEE Internet of Things Journal}, vol.~5, no.~1, pp.
  450--465, 2017.

\bibitem{liu2020resource}
B.~Liu, C.~Liu, and M.~Peng, ``Resource allocation for energy-efficient mec in
  noma-enabled massive iot networks,'' \emph{IEEE Journal on Selected Areas in
  Communications}, 2020.

\bibitem{du2018energy}
Y.~Du, K.~Wang, K.~Yang, and G.~Zhang, ``Energy-efficient resource allocation
  in uav based mec system for iot devices,'' in \emph{2018 IEEE Global
  Communications Conference (GLOBECOM)}.\hskip 1em plus 0.5em minus 0.4em\relax
  IEEE, 2018, pp. 1--6.

\bibitem{8767017}
X.~{Hu}, K.~{Wong}, K.~{Yang}, and Z.~{Zheng}, ``Uav-assisted relaying and edge
  computing: Scheduling and trajectory optimization,'' \emph{IEEE Transactions
  on Wireless Communications}, vol.~18, no.~10, pp. 4738--4752, 2019.

\bibitem{mach2017mobile}
P.~Mach and Z.~Becvar, ``Mobile edge computing: A survey on architecture and
  computation offloading,'' \emph{IEEE Communications Surveys \& Tutorials},
  vol.~19, no.~3, pp. 1628--1656, 2017.

\bibitem{zhao2019novel}
Z.~Zhao, R.~Zhao, J.~Xia, X.~Lei, D.~Li, C.~Yuen, and L.~Fan, ``A novel
  framework of three-hierarchical offloading optimization for mec in industrial
  iot networks,'' \emph{IEEE Transactions on Industrial Informatics}, vol.~16,
  no.~8, pp. 5424--5434, 2019.

\bibitem{zhao2019computation}
J.~Zhao, Q.~Li, Y.~Gong, and K.~Zhang, ``Computation offloading and resource
  allocation for cloud assisted mobile edge computing in vehicular networks,''
  \emph{IEEE Transactions on Vehicular Technology}, vol.~68, no.~8, pp.
  7944--7956, 2019.

\bibitem{loven2019edgeai}
L.~Lov{\'e}n, T.~Lepp{\"a}nen, E.~Peltonen, J.~Partala, E.~Harjula,
  P.~Porambage, M.~Ylianttila, and J.~Riekki, ``Edgeai: A vision for
  distributed, edgenative artificial intelligence in future 6g networks,''
  \emph{The 1st 6G Wireless Summit}, pp. 1--2, 2019.

\bibitem{mnih2015human}
V.~Mnih, K.~Kavukcuoglu, D.~Silver, A.~A. Rusu, J.~Veness, M.~G. Bellemare,
  A.~Graves, M.~Riedmiller, A.~K. Fidjeland, G.~Ostrovski \emph{et~al.},
  ``Human-level control through deep reinforcement learning,'' \emph{nature},
  vol. 518, no. 7540, pp. 529--533, 2015.

\bibitem{ye2020mastering}
D.~Ye, Z.~Liu, M.~Sun, B.~Shi, P.~Zhao, H.~Wu, H.~Yu, S.~Yang, X.~Wu, Q.~Guo
  \emph{et~al.}, ``Mastering complex control in moba games with deep
  reinforcement learning.'' in \emph{AAAI}, 2020, pp. 6672--6679.

\bibitem{sutton2018reinforcement}
R.~S. Sutton and A.~G. Barto, \emph{Reinforcement learning: An
  introduction}.\hskip 1em plus 0.5em minus 0.4em\relax MIT press, 2018.

\bibitem{lillicrap2015continuous}
T.~P. Lillicrap, J.~J. Hunt, A.~Pritzel, N.~Heess, T.~Erez, Y.~Tassa,
  D.~Silver, and D.~Wierstra, ``Continuous control with deep reinforcement
  learning,'' \emph{arXiv preprint arXiv:1509.02971}, 2015.

\bibitem{busoniu2008comprehensive}
L.~Busoniu, R.~Babuska, and B.~De~Schutter, ``A comprehensive survey of
  multiagent reinforcement learning,'' \emph{IEEE Transactions on Systems, Man,
  and Cybernetics, Part C (Applications and Reviews)}, vol.~38, no.~2, pp.
  156--172, 2008.

\bibitem{zhuo2019federated}
H.~H. Zhuo, W.~Feng, Q.~Xu, Q.~Yang, and Y.~Lin, ``Federated reinforcement
  learning,'' \emph{arXiv preprint arXiv:1901.08277}, 2019.

\bibitem{ndikumana2019joint}
A.~Ndikumana, N.~H. Tran, T.~M. Ho, Z.~Han, W.~Saad, D.~Niyato, and C.~S. Hong,
  ``Joint communication, computation, caching, and control in big data
  multi-access edge computing,'' \emph{IEEE Transactions on Mobile Computing},
  vol.~19, no.~6, pp. 1359--1374, 2019.

\bibitem{merwaday2015uav}
A.~Merwaday and I.~Guvenc, ``Uav assisted heterogeneous networks for public
  safety communications,'' in \emph{2015 IEEE wireless communications and
  networking conference workshops (WCNCW)}.\hskip 1em plus 0.5em minus
  0.4em\relax IEEE, 2015, pp. 329--334.

\bibitem{sharma2016uav}
V.~Sharma, M.~Bennis, and R.~Kumar, ``Uav-assisted heterogeneous networks for
  capacity enhancement,'' \emph{IEEE Communications Letters}, vol.~20, no.~6,
  pp. 1207--1210, 2016.

\bibitem{emara2020spatiotemporal}
M.~Emara, H.~El~Sawy, M.~C. Filippou, and G.~Bauch, ``Spatiotemporal dependable
  task execution services in mec-enabled wireless systems,'' \emph{IEEE
  Wireless Communications Letters}, 2020.

\bibitem{cao2019intelligent}
B.~Cao, L.~Zhang, Y.~Li, D.~Feng, and W.~Cao, ``Intelligent offloading in
  multi-access edge computing: A state-of-the-art review and framework,''
  \emph{IEEE Communications Magazine}, vol.~57, no.~3, pp. 56--62, 2019.

\bibitem{ning2020partial}
Z.~Ning, P.~Dong, X.~Wang, X.~Hu, J.~Liu, L.~Guo, B.~Hu, R.~Kwok, and V.~C.
  Leung, ``Partial computation offloading and adaptive task scheduling for
  5g-enabled vehicular networks,'' \emph{IEEE Transactions on Mobile
  Computing}, 2020.

\bibitem{matolak2015unmanned}
D.~W. Matolak and R.~Sun, ``Unmanned aircraft systems: Air-ground channel
  characterization for future applications,'' \emph{IEEE Vehicular Technology
  Magazine}, vol.~10, no.~2, pp. 79--85, 2015.

\bibitem{zhou2018uav}
F.~Zhou, Y.~Wu, H.~Sun, and Z.~Chu, ``Uav-enabled mobile edge computing:
  Offloading optimization and trajectory design,'' in \emph{2018 IEEE
  International Conference on Communications (ICC)}.\hskip 1em plus 0.5em minus
  0.4em\relax IEEE, 2018, pp. 1--6.

\bibitem{liu2019uav}
Y.~Liu, K.~Xiong, Q.~Ni, P.~Fan, and K.~B. Letaief, ``Uav-assisted wireless
  powered cooperative mobile edge computing: Joint offloading, cpu control, and
  trajectory optimization,'' \emph{IEEE Internet of Things Journal}, vol.~7,
  no.~4, pp. 2777--2790, 2019.

\bibitem{liu2017latency}
C.-F. Liu, M.~Bennis, and H.~V. Poor, ``Latency and reliability-aware task
  offloading and resource allocation for mobile edge computing,'' in \emph{2017
  IEEE Globecom Workshops (GC Wkshps)}.\hskip 1em plus 0.5em minus 0.4em\relax
  IEEE, 2017, pp. 1--7.

\bibitem{costa2014age}
M.~Costa, M.~Codreanu, and A.~Ephremides, ``Age of information with packet
  management,'' in \emph{2014 IEEE International Symposium on Information
  Theory}.\hskip 1em plus 0.5em minus 0.4em\relax IEEE, 2014, pp. 1583--1587.

\bibitem{wang2021minimizing}
X.~Wang, Z.~Ning, S.~Guo, M.~Wen, and V.~Poor, ``Minimizing the
  age-of-critical-information: An imitation learning-based scheduling approach
  under partial observations,'' \emph{IEEE Transactions on Mobile Computing},
  2021.

\bibitem{chen2020age}
X.~Chen, C.~Wu, T.~Chen, H.~Zhang, Z.~Liu, Y.~Zhang, and M.~Bennis, ``Age of
  information aware radio resource management in vehicular networks: A
  proactive deep reinforcement learning perspective,'' \emph{IEEE Transactions
  on Wireless Communications}, vol.~19, no.~4, pp. 2268--2281, 2020.

\bibitem{liu2018age}
J.~Liu, X.~Wang, B.~Bai, and H.~Dai, ``Age-optimal trajectory planning for
  uav-assisted data collection,'' in \emph{IEEE INFOCOM 2018-IEEE Conference on
  Computer Communications Workshops (INFOCOM WKSHPS)}.\hskip 1em plus 0.5em
  minus 0.4em\relax IEEE, 2018, pp. 553--558.

\bibitem{hu2020aoi}
H.~Hu, K.~Xiong, G.~Qu, Q.~Ni, P.~Fan, and K.~B. Letaief, ``Aoi-minimal
  trajectory planning and data collection in uav-assisted wireless powered iot
  networks,'' \emph{IEEE Internet of Things Journal}, 2020.

\bibitem{tong2020deep}
P.~Tong, J.~Liu, X.~Wang, B.~Bai, and H.~Dai, ``Deep reinforcement learning for
  efficient data collection in uav-aided internet of things,'' in \emph{2020
  IEEE International Conference on Communications Workshops (ICC
  Workshops)}.\hskip 1em plus 0.5em minus 0.4em\relax IEEE, 2020, pp. 1--6.

\bibitem{wang2019smart}
J.~Wang, L.~Zhao, J.~Liu, and N.~Kato, ``Smart resource allocation for mobile
  edge computing: A deep reinforcement learning approach,'' \emph{IEEE
  Transactions on emerging topics in computing}, 2019.

\bibitem{li2018deep}
J.~Li, H.~Gao, T.~Lv, and Y.~Lu, ``Deep reinforcement learning based
  computation offloading and resource allocation for mec,'' in \emph{2018 IEEE
  Wireless Communications and Networking Conference (WCNC)}.\hskip 1em plus
  0.5em minus 0.4em\relax IEEE, 2018, pp. 1--6.

\bibitem{wan2019towards}
S.~Wan, J.~Lu, P.~Fan, and K.~B. Letaief, ``Towards big data processing in iot:
  Path planning and resource management of uav base stations in mobile-edge
  computing system,'' \emph{IEEE Internet of Things Journal}, 2019.

\bibitem{chen2018optimized}
X.~Chen, H.~Zhang, C.~Wu, S.~Mao, Y.~Ji, and M.~Bennis, ``Optimized computation
  offloading performance in virtual edge computing systems via deep
  reinforcement learning,'' \emph{IEEE Internet of Things Journal}, vol.~6,
  no.~3, pp. 4005--4018, 2018.

\bibitem{wang2020multi1}
X.~Wang, Z.~Ning, and S.~Guo, ``Multi-agent imitation learning for pervasive
  edge computing: a decentralized computation offloading algorithm,''
  \emph{IEEE Transactions on Parallel and Distributed Systems}, vol.~32, no.~2,
  pp. 411--425, 2020.

\bibitem{peng2020multi}
H.~Peng and X.~Shen, ``Multi-agent reinforcement learning based resource
  management in mec-and uav-assisted vehicular networks,'' \emph{IEEE Journal
  on Selected Areas in Communications}, 2020.

\bibitem{wang2020multi}
L.~Wang, K.~Wang, C.~Pan, W.~Xu, N.~Aslam, and L.~Hanzo, ``Multi-agent deep
  reinforcement learning based trajectory planning for multi-uav assisted
  mobile edge computing,'' \emph{IEEE Transactions on Cognitive Communications
  and Networking}, 2020.

\bibitem{zhang2020uav}
Y.~Zhang, Z.~Mou, F.~Gao, J.~Jiang, R.~Ding, and Z.~Han, ``Uav-enabled secure
  communications by multi-agent deep reinforcement learning,'' \emph{IEEE
  Transactions on Vehicular Technology}, vol.~69, no.~10, pp. 11\,599--11\,611,
  2020.

\bibitem{kumar2017federated}
S.~Kumar, P.~Shah, D.~Hakkani-Tur, and L.~Heck, ``Federated control with
  hierarchical multi-agent deep reinforcement learning,'' \emph{arXiv preprint
  arXiv:1712.08266}, 2017.

\bibitem{wang2020federated}
X.~Wang, C.~Wang, X.~Li, V.~C. Leung, and T.~Taleb, ``Federated deep
  reinforcement learning for internet of things with decentralized cooperative
  edge caching,'' \emph{IEEE Internet of Things Journal}, 2020.

\bibitem{liu2016delay}
J.~Liu, Y.~Mao, J.~Zhang, and K.~B. Letaief, ``Delay-optimal computation task
  scheduling for mobile-edge computing systems,'' in \emph{2016 IEEE
  International Symposium on Information Theory (ISIT)}.\hskip 1em plus 0.5em
  minus 0.4em\relax IEEE, 2016, pp. 1451--1455.

\bibitem{al2014optimal}
A.~Al-Hourani, S.~Kandeepan, and S.~Lardner, ``Optimal lap altitude for maximum
  coverage,'' \emph{IEEE Wireless Communications Letters}, vol.~3, no.~6, pp.
  569--572, 2014.

\bibitem{kosta2017age}
A.~Kosta, N.~Pappas, and V.~Angelakis, ``Age of information: A new concept,
  metric, and tool,'' \emph{Foundations and Trends in Networking}, vol.~12,
  no.~3, pp. 162--259, 2017.

\bibitem{costa2016age}
M.~Costa, M.~Codreanu, and A.~Ephremides, ``On the age of information in status
  update systems with packet management,'' \emph{IEEE Transactions on
  Information Theory}, vol.~62, no.~4, pp. 1897--1910, 2016.

\bibitem{wiering2012reinforcement}
M.~Wiering and M.~Van~Otterlo, \emph{Reinforcement learning}.\hskip 1em plus
  0.5em minus 0.4em\relax Springer, 2012, vol.~12.

\bibitem{van2012reinforcement}
M.~Van~Otterlo and M.~Wiering, ``Reinforcement learning and markov decision
  processes,'' in \emph{Reinforcement Learning}.\hskip 1em plus 0.5em minus
  0.4em\relax Springer, 2012, pp. 3--42.

\bibitem{littman1994markov}
M.~L. Littman, ``Markov games as a framework for multi-agent reinforcement
  learning,'' in \emph{Machine learning proceedings 1994}.\hskip 1em plus 0.5em
  minus 0.4em\relax Elsevier, 1994, pp. 157--163.

\bibitem{panait2005cooperative}
L.~Panait and S.~Luke, ``Cooperative multi-agent learning: The state of the
  art,'' \emph{Autonomous agents and multi-agent systems}, vol.~11, no.~3, pp.
  387--434, 2005.

\bibitem{puterman2014markov}
M.~L. Puterman, \emph{Markov decision processes: discrete stochastic dynamic
  programming}.\hskip 1em plus 0.5em minus 0.4em\relax John Wiley \& Sons,
  2014.

\bibitem{lowe2017multi}
R.~Lowe, Y.~I. Wu, A.~Tamar, J.~Harb, O.~P. Abbeel, and I.~Mordatch,
  ``Multi-agent actor-critic for mixed cooperative-competitive environments,''
  in \emph{Advances in neural information processing systems}, 2017, pp.
  6379--6390.

\bibitem{arulkumaran2017brief}
K.~Arulkumaran, M.~P. Deisenroth, M.~Brundage, and A.~A. Bharath, ``A brief
  survey of deep reinforcement learning,'' \emph{arXiv preprint
  arXiv:1708.05866}, 2017.

\bibitem{yang2019federated}
Q.~Yang, Y.~Liu, T.~Chen, and Y.~Tong, ``Federated machine learning: Concept
  and applications,'' \emph{ACM Transactions on Intelligent Systems and
  Technology (TIST)}, vol.~10, no.~2, pp. 1--19, 2019.

\bibitem{konevcny2016federated}
J.~Kone{\v{c}}n{\`y}, H.~B. McMahan, F.~X. Yu, P.~Richt{\'a}rik, A.~T. Suresh,
  and D.~Bacon, ``Federated learning: Strategies for improving communication
  efficiency,'' \emph{arXiv preprint arXiv:1610.05492}, 2016.

\bibitem{wang2018cooperative}
J.~Wang and G.~Joshi, ``Cooperative sgd: A unified framework for the design and
  analysis of communication-efficient sgd algorithms,'' \emph{arXiv preprint
  arXiv:1808.07576}, 2018.

\bibitem{balan2017lipschitz}
R.~Balan, M.~Singh, and D.~Zou, ``Lipschitz properties for deep convolutional
  networks,'' \emph{arXiv preprint arXiv:1701.05217}, 2017.

\bibitem{latorre2020lipschitz}
F.~Latorre, P.~Rolland, and V.~Cevher, ``Lipschitz constant estimation of
  neural networks via sparse polynomial optimization,'' \emph{arXiv preprint
  arXiv:2004.08688}, 2020.

\bibitem{haddadpour2019convergence}
F.~Haddadpour and M.~Mahdavi, ``On the convergence of local descent methods in
  federated learning,'' \emph{arXiv preprint arXiv:1910.14425}, 2019.

\bibitem{brockman2016openai}
G.~Brockman, V.~Cheung, L.~Pettersson, J.~Schneider, J.~Schulman, J.~Tang, and
  W.~Zaremba, ``Openai gym,'' \emph{arXiv preprint arXiv:1606.01540}, 2016.

\bibitem{fan2018application}
Q.~Fan and N.~Ansari, ``Application aware workload allocation for edge
  computing-based iot,'' \emph{IEEE Internet of Things Journal}, vol.~5, no.~3,
  pp. 2146--2153, 2018.

\end{thebibliography}

\end{document}